%% file: convergence_new.tex
\documentclass[10pt]{article} 
\usepackage{graphicx}
\usepackage{amsmath, amsthm}
\usepackage{amsmath}
\usepackage{amsfonts}
\usepackage{amssymb}
\usepackage{mathrsfs}
\usepackage{verbatim}
\usepackage{wrapfig}
\usepackage{pgf}
\usepackage{subfigure}
\usepackage{color}
\usepackage{setspace}

\newtheorem*{mainthm}{Main Theorem}
\newtheorem{thm}{Theorem}[section]
\newtheorem{cor}[thm]{Corollary}
\newtheorem{lemma}[thm]{Lemma}
\newtheorem{prop}[thm]{Proposition}

\theoremstyle{definition}
\newtheorem{remark}[thm] {Remark}

\newtheorem{defn}[thm]{Definition}
\newtheorem{ex}[thm]{Example}

\newcommand{\e}{\varepsilon}

\begin{document} 
\title{\Large{On Convergence to SLE$_6$ I:\\ \large{Conformal Invariance for Certain Models of the Bond--Triangular Type}}}
\date{}
\author
{\normalsize{I.~Binder$^1$, L.~Chayes$^2$, H.~K.~Lei$^2$}
\thanks{\copyright\, 2010 by I. Binder, L.~Chayes and H.~K.~Lei.  Reproduction, by any means, of the entire article for non-commercial purposes is permitted without charge.}}
\maketitle

\vspace{-4mm}
\centerline{${}^1$\textit{Department of Mathematics, University of Toronto}}
\centerline{${}^2$\textit{Department of Mathematics, UCLA}}
 
\begin{quote}
{\footnotesize {\bf Abstract: }}Following the approach outlined in \cite{stas}, convergence to SLE$_6$ of the Exploration Processes for the correlated bond--triangular type models studied in \cite{cardy} is established.   This puts the said models in the same universality class as the standard site percolation model on the triangular lattice \cite{stas_perc}.  In the context of these models, the result is proven for all domains with boundary Minkowski dimension less than two. Moreover, the proof of convergence applies in the context of general critical 2D percolation models and for general domains, under the stipulation that Cardy's Formula can be established for domains in this generality.  

{\footnotesize {\bf Keywords: }}Universality, conformal invariance, percolation, Cardy's Formula.
\end{quote}


\section{\large{Introduction}}
In recent years, the scaling behavior of critical 2D percolation systems have been the subject of attention.  While the results proved in this note amount to a statement concerning the scaling limit of the specific percolation models defined in \cite{cardy}, the purpose of this work is actually three--fold: 1) Following the framework described in \cite{stas}, we provide a general proof that (the law of) the ``interface'' of essentially \emph{any} critical 2D percolation model converges to SLE$_6$, whenever Cardy's Formula can be verified.  2) Rigorous extraction of Cardy's Formula for general domains -- including slit domains, given interior analyticity of the Cardy--Carleson functions; this includes clarification of the necessary discretization schemes.  3) Finally, we provide a generalization of Cardy's Formula to an extended class of domains for the specific class of models described in \cite{cardy}, and also establish additional ``typical'' (critical) percolation properties which are required, in accord with 1) and 2) above.  We accomplish 1) and 3) in the current installment of this work; item 2) will be tended to in a separate (companion) note \cite{pt2}.  
%

It is already well--known \cite{stas_perc} that site percolation on the 2D triangular lattice satisfies these sorts of properties.  While in \cite{cn} an elaborate proof of convergence to SLE$_6$ has been detailed, and while it is possible that the proof therein applies in more generality than claimed, the present approach is manifestly applicable to a variety of systems and in a variety of domains.  As a result we have, in complete accordance with the ideology espoused since the 1960s, demonstrated a non--trivial example of universality: Via the common continuum limit, various aspects of the long distance behavior for the models defined in \cite{cardy} are asymptotically identical to those of the critical triangular site percolation model.   

We remark that in principle, our proof applies in the general context of any critical 2D percolation model.  The required conditions are summarized as follows:
\begin{itemize}
\item Russo--Seymour--Welsh (RSW) theory: Uniform estimates for probabilities of crossings (of either type) on all scales plus the ability to stitch smaller crossings together without substantial degradation of the estimates -- FKG--type inequalities.
\item A self--replicating definition of an Exploration Process and a class of \emph{admissible} domains with the property that this class is preserved under the operation of deleting the beginning of a typical explorer path in an admissible domain.
\item The validity of Cardy's Formula for the above--mentioned admissible domains.
\item BK--type inequalities whereby probabilities of separated path type events can be estimated in terms of the individual probabilities.
\item Explicit (``superuniversal'') ``bounds'' on full--space multiple colored five--arm events and half--space multiple colored three--arm events: The probability of observing disjoint crossings of an annulus with aspect ratio $a$ is, on all scales, bounded above by a constant times $a^{-2}$.
\end{itemize}

The rest of this paper is organized as follows: In Section \ref{proof_convergence}, we assemble the necessary ingredients into the proof of convergence to SLE$_6$ (providing some minor proofs of an analytical nature along the way).  These ingredients amount to a number of technical lemmas, a few of which require a sustained effort and whose proofs are provided in Section \ref{proofs} and, for one of them, a result imported from \cite{pt2}.  
Finally, Section \ref{model_specific} is devoted to shoring up the required properties of the models defined in \cite{cardy} to the appropriate level for the program in Section \ref{proof_convergence}.


\section{Conformal Invariance of the Scaling Limit}\label{proof_convergence}
\subsection{2D Percolation:  Criticality and  Interfaces (a Brief Discussion)}
In this subsection, we shall elucidate, to some extent, the first and second (bullet) items in the penultimate paragraph of the introduction.  For brevity -- and purposes of clarity -- we will not attempt to axiomatize the relevant notions.  In general, the percolation process consists of two competing species, conveniently denoted by ``blue'' and ``yellow''.  The condition of criticality implies that the two species have roughly equal parity; it need not be the case that the two are exactly equivalent, but neither species is dominant at large scales.  In particular, there is no percolation of either species -- with probability one, all monochrome connected clusters are finite.  As it turns out, this is (more or less) equivalent to the statement that for both species, at all scales, the probability of crossing ``rectangles'' of fixed ratio is bounded above and below uniformly.  Moreover, with some notion of positive correlations for crossing type events of the same color, we may patch together the appropriate crossings to conclude that there are scale--invariant bounds on the existence of circuits in annuli; since \emph{Bernoulli percolation} is supposed to imply independence beyond some fixed scale, this also implies similar estimates for circuits in ``partial annuli'' and approximate independence in disjoint layered annuli.  Typically, the way such estimates are applied is as follows: There is a large outside scale and a small inside scale separated by logarithmically many intermediate scales; the probability of monochrome connections between the inner and outer scale is therefore a power of the ratio.  This is the basis of the so--called Russo--Seymour--Welsh (RSW) theory which will be used throughout this work.  For the standard percolation models, these concepts are discussed in the books \cite{kesten_book}, \cite{grimmett_book} and \cite{br}; see also Sections 2.2 and 2.3 of \cite{les_houches} and the paper  \cite{kesten_power}.  For the particular model of interest in this work, such results are not quite automatic, but anyway have been established in \cite{cardy}, the relevant portions of which will be cited as necessary.    

In a similar spirit, let us now discuss critical interfaces for these models (although strictly speaking, criticality plays no r\^ole).  The general setup is as follows: For any finite connected lattice domain, let us fix two ``boundary points'' $a$ and $c$ and impose boundary conditions so that the portion of the boundary going from $a$ to $c$ one way is colored blue and the complementary portion of the boundary is yellow.  The precise lattice--mechanics depend, of course, on the model at hand (and indeed may involve different procedures on the yellow and blue sides).  In any case, if this procedure has been implemented successfully, then in any percolation configuration there will be an \emph{interface} stretching from $a$ to $c$, which separates the blue connected component of the blue boundary from the yellow component of the yellow boundary.  The explicit construction for our model will be provided in Section \ref{explorer}; well known examples include the triangular site percolation problem and the bond model on $\mathbb Z^2$.  In the former case, 
the interface can be realized as boundary segments of hexagons and in the latter, interface consists of segments which connect sites of the so--called medial lattice.  

The seminal ingredient is the Domain Markov Property: The full percolation model with the above boundary setup conditioned on an initial portion of the Exploration Process is identical to the problem in the ``slit'' domain with additional (two--colored) boundary formed by the corresponding curve segment.  It seems manifest, at least for planar models, that all 2D percolation systems have this property.  Whereas the preceding may seem rather vague and discursive, what is actually needed is somewhat less and succinctly formulated: The precise requirement is the content of Equation \eqref{eq:tot}, which is the restriction of these notions to crossing events.  

\subsection{SLE: Definitions and Notations}
\label{SLE}
As the title of this subsection indicates, we will briefly review the relevant notions of L\"owner evolution -- mostly for the purpose of fixing notation.  Let $\Omega$ be a domain with two boundary prime ends $a$ and $c$.
\begin{defn}
Let $\{\Omega_t\}_{t=0}^\infty$ be a strictly decreasing family of subdomains of $\Omega$ ($t \in [0, \infty)$) which is Carath\'eodory continuous with respect to $c$, such that $\Omega_0 = \Omega$ and $c \in \cap_{t=0}^\infty \overline{\Omega}_t$.  Then we call $\{\Omega_t\}_{t=0}^\infty$ a L\"owner chain.
\end{defn}

Let $\mathbb H$ denote the upper--half plane of $\mathbb C$.  We can select some conformal map $g_0: \Omega \rightarrow \mathbb H$ such that $g_0(a) = 0$ and $g_0(c) = \infty$.  The family of conformal maps $g_t: \Omega_t \rightarrow \mathbb H$ normalized such that  $g_t(c) = \infty$ and $g_t \circ g_0^{-1}(z) = z + \frac{A(t)}{z} + o(1/z)$ are continuous in $t$.  We now reparameterize time so that $A(t)$, the capacity at time $t$, is equal to $2t$.  

We call $\gamma$ a \emph{crosscut} in $\Omega$ from $a$ to $c$ if it is the preimage of a non--self--crossing curve from $0$ to $\infty$ in $\mathbb H$ under $g_0$.  Note that  $\gamma$ is allowed to touch itself but not to cross itself.  We define $\Omega_t$ to be the connected component of $\Omega \setminus \gamma_{[0, t]}$ containing $c$.  It's easy to see that $\Omega_t$ is a L\"owner chain if and only if the following two conditions are satisfied for every $t > 0$:\\
 (L1)~~$\gamma_t \in \overline{\Omega_{t-\varepsilon}}, \quad \forall \varepsilon > 0$\\
and\\
 (L2) ~~$\exists \delta_n \rightarrow 0, \quad \forall \varepsilon > 0, \quad \gamma_{t-\delta_n} \in \Omega_{t - \delta_n - \varepsilon}$.\\
If $\gamma$ satisfies (L1) and (L2), then we say that $\gamma$ is a \emph{L\"owner curve}.
Under these conditions, we can reparametrize $\gamma$ so that the maps $g_t$'s satisfy the following celebrated L\"owner equation:
\[ \partial_t g_t (z) = \frac{2}{g_t(z) - \lambda_t}, \]
where $\lambda_t=g_t(\gamma(t))$ is a continuous real function.  On the other hand, the solution of the L\"owner equation for any initial conformal map $g_0: \Omega \rightarrow \mathbb H$ and any continuous real function $\lambda(t)$ defines a L\"owner chain, but not necessarily a curve (see \cite{lawler} for a complete discussion). The object $\lambda_t$ is called the {\it driving function} of $\Omega_t$.

If we take the very special function $\lambda_t=B(\kappa t)$, where $B(t)$ is one-dimension Brownian motion started at zero, then the corresponding random L\"owner chain is called the Stochastic (or Schramm) L\"owner Evolution with parameter $\kappa$, SLE${}_\kappa$.  
We will be particularly interested in the case $\kappa=6$.

\subsection{Statement of the Main Theorem and Lemmas}\label{statements}

We start with a bounded and connected domain $\Omega \subset \mathbb C$.  We will sometimes assume that $\Omega$ has ``boundary dimension'' $M(\partial \Omega) < 2$.  Here $M(S)$ denotes the (upper) \emph{Minkowski} dimension of the set $S$ which, as usual, is defined as 
\[ M (S) = \limsup_{\vartheta \rightarrow 0} \frac{\log \mathcal N(\vartheta)}{\log (1/\vartheta)},\]
where $\mathcal N(\vartheta)$ is the number of boxes of side length $\vartheta$ needed to cover the set.  We will tile $\Omega$ with the discrete lattice of interest (which may require detail, c.f.~$\S$\ref{explorer} and, especially, the discussion in \cite{pt2}) at scale $\varepsilon > 0$ and denote the resulting object by $\Omega_\varepsilon$.  Critical percolation is then performed in $\Omega_\varepsilon$, with $\varepsilon$ tending to zero.  

While the principal result of this note has more general applicability, for simplicity let us state it for the particular model under consideration: 

\begin{mainthm}
Let $\Omega$ be as described above with $M(\partial \Omega) < 2$, let $\Omega_\varepsilon$ be some suitable discretization (see \cite{pt2} for discussions and results) and consider the percolation model described in \cite{cardy} (see $\S$\ref{model}).  Let $a$ and $c$ denote two prime ends at the boundary of $\Omega$ and let us set the boundary conditions on $\Omega_\varepsilon$ in such a way that the Exploration Process, as defined in $\S$\ref{explorer}, runs between $a$ and $c$.  Let $\mu_\varepsilon$ be the probability measure on random curves induced by the Exploration Process on $\Omega_\varepsilon$, and let us endow the space of curves with the appropriate weighted sup--norm metric as described in Definition \ref{frechet}.  Then, 
\[ \mu_\varepsilon \underset{\mathcal{L~}}{\Longrightarrow} \mu_0,\]
where $\mu_0$ has the law of chordal SLE$_6$ from $a$ to $c$.
\end{mainthm}

We remark that while the above statement appears to require a number of ``future specifics'', these are merely technicalities.  The central requisites are captured in the items listed in the penultimate paragraph of the introduction and will be detailed as the proof of the Main Theorem unfolds.  (In particular, here and throughout, the requirement $M(\partial \Omega) < 2$ is for the specific benefit of the model defined in \cite{cardy}.)

The key ingredient which will be used in the proof of the Main Theorem is Cardy's Formula:

\begin{lemma}[Cardy's Formula]\label{cardy_formula}
Let $(\Omega,a,b,c,d)$ be a conformal rectangle -- that is to say, a domain with boundary prime ends $a,b,c,d$, listed in counter-clockwise order, and let us assume that
$M(\partial \Omega)<2$. Let $C_{\varepsilon}(\Omega,a,b,c,d)$ denote the probability that there exists a blue crossing from $[a,b]$ to $[c,d]$ on the $\varepsilon$-lattice approximation of $\Omega$.  Consider the (unique) conformal map which takes $(\Omega, a, b, c, d)$ to $({\mathbb H}, 1-x, 1, \infty, 0)$, where, clearly, $0 < x < 1$ and $x = x(\Omega, a, b, c, d)$.  Then, for the model described in $\S$\ref{model} (or without the restriction $M(\partial \Omega) < 2$ for the site percolation model)
\begin{equation}\label{eq:cardy}
\lim_{\varepsilon\to0}C_{\varepsilon}(\Omega,a,b,c,d)=F(x):=\frac{\int_0^x(s(1-s))^{-2/3}\,ds}{\int_0^1(s(1-s))^{-2/3}\,ds}.
\end{equation}
\end{lemma}
\begin{proof}
This, modulo the formula \eqref{eq:cardy}, is the content of \cite{pt2}, Theorem 4.7.  For the particular model at hand, this was established for a restricted class of domains in \cite{cardy}.  The necessary generalization of the work in \cite{cardy} to domains with $M(\partial \Omega) < 2$ will be proved in $\S$\ref{cardy} (see Lemma \ref{cardy_formula'}). 
\end{proof}

Using general estimates in $\S$\ref{explorer_properties}, we establish the following important properties of any weak$^*$--limiting point $\mu^\prime$.  The proofs can be found in $\S$\ref{aizenman} and $\S$\ref{preservation}.

\begin{lemma}[Tightness]\label{tightness}
Let $\mu^\prime$ be any limit point, in the weak$^*$ Hausdorff topology on compact sets, of $\mu_\varepsilon$. Then $\mu^\prime$ gives full measure to L\"{o}wner curves in $\Omega$ from $a$ to $c$.  
\end{lemma}

Furthermore, we have

\begin{lemma}[Admissibility]\label{admissibility}
The limit point $\mu^\prime$ gives full measure to curves with upper Minkowski dimension less than $2 - \psi^\prime$ for some $\psi^\prime > 0$.
\end{lemma}

We note that in Lemma \ref{tightness} (and Lemma \ref{admissibility}), a stronger notion of convergence is available.  Indeed, for domains which are regular enough, the results of \cite{ab} provide weak$^*$ convergence to $\mu^{\prime}$ in the distance provided by the sup--norm:
\[ \mbox{dist}(\gamma_1, \gamma_2) = \inf_{\varphi_1, \varphi_2} \sup_t |\gamma_1(\varphi_1(t)) - \gamma_2(\varphi_2(t))|,\]
where the infimum is over all possible parametrizations.  For our purposes -- where prime ends are a concern -- we will consider a weighted sum of the distances within various regions between the curves.  We will denote the appropriate distance by $\mathbf{Dist}$; see Definition \ref{frechet}.  We can easily extend the result of \cite{ab} to the following:

\begin{lemma}[\textbf{Dist} Topology]\label{Dist_conv}
The measure $\mu^\prime$ is a limit point in the weak$^*$ \textbf{\emph{Dist}} topology on curves of $\mu_\varepsilon$. 
\end{lemma}

Finally, we will use the following continuity result for crossing probabilities, whose proof can be found in \cite{pt2} (stated as Corollary 5.10):

\begin{lemma}\label{pt_eq_cont}
Consider the models described in \cite{cardy} (which includes the triangular site problem studied in \cite{stas_perc}) on a bounded domain $\Omega$ with boundary Minkowski dimension less than two (if necessary) and two marked boundary points $a$ and $c$.  Consider $\mathscr C_{a, c, \Delta}$, the set of L\"oewner curves which begin at the point $a$ and are aiming towards the point $c$ but have not yet entered the $\Delta$ neighborhood of $c$ for some $\Delta > 0$.  Suppose we have $\gamma^\e \rightarrow \gamma$ in the \textbf{\emph{Dist}} norm, then 
\[C_\e(\Omega_\e \setminus \gamma_\e([0, t]), \gamma_\e(t), b_\e, c_\e, d_\e) \rightarrow C_0(\Omega \setminus \gamma([0, t]), \gamma(t), b, c, d)\]
pointwise equicontinuously in the sense that
\begin{equation}\label{pt_equi_cont}\begin{array}{c} \forall \sigma > 0, ~~\forall \gamma \in \mathscr C_{a, c, \Omega}, ~~\exists \delta(\gamma) > 0, ~~\exists \mathcal \e_\gamma,\\ 
\\
\mbox{ such that }\\ 
\\
\forall \gamma^\prime \in \mathcal B_{\delta(\gamma)}(\gamma),~~ \forall \e \leq \mathcal \e_\gamma,\\ \\  
|C_\e((\Omega \setminus \gamma)_\e([0, t])), (\gamma(t))_\e, b_\e, c_\e, d_\e) - C_\e((\Omega \setminus \gamma^\prime)_\e([0, t])), (\gamma^\prime(t))_\e, b_\e, c_\e, d_\e)| < \sigma.\end{array}\end{equation}
Here $B_\delta(\gamma)$ denotes the \textbf{\emph{Dist}} neighborhood of $\gamma$.  
\end{lemma}

\subsection{Proof of the Main Theorem}\label{smir_prog}

Let us show how to derive our Main Theorem from the preceding lemmas. We closely follow the strategic initiative outlined in the expositions of \cite{stas} (for a slightly different and more probabilistic perspective on the subject, also see the exposition in \cite{werner_perc}); moreover, the ``expansion at infinity'' technique we will use here first appeared in \cite{lsw} in the proof of the convergence of the loop--erased random walk to SLE$_2$.

Let us fix $\Omega$ with $M(\partial\Omega)<2$ and two boundary prime ends $a$ and $c$. We start with an informal list of the key steps.

\begin{itemize}
\item[I.] Extract some limiting measure $\mu^\prime$. 
\item[II.]Show that any limiting measure is supported on L\"oewner curves.
\item[III.] Establish the discrete domain (crossing) Markov property.
\item[IV.] L\"oewner parameterize all curves under consideration.
\item[V.] Obtain the limiting martingale.  
\item[VI.] Show that $\kappa = 6$.
\end{itemize}

$\diamond$ I.\hspace{0.5mm}]  Let us note that the collection of measures $(\mu_{\varepsilon})$ defined by the Exploration Processes on $\varepsilon$-lattice is weakly precompact as a set of regular measures defined on the space of compact subsets of $\overline \Omega$ with the Hausdorff metric. Thus to prove the Main Theorem it is enough to show that any weak limit point $\mu^\prime$, of $\mu_{\varepsilon}$, has the law of SLE${}_6$ from $a$ to $c$ in $\Omega$.

\vspace{0.5cm}

$\diamond$ II.\hspace{0.5mm}]  By Lemma \ref{tightness}, $\mu'$ gives full measure to L\"owner curves.  Let $w_t$ be the random driving function of the curve.  To finish the proof, we need to show that $w_t$ has the law of $B_{6t}$, where $B_t$ is the standard one dimensional Brownian Motion started at $0$. 

\vspace{0.5cm}

$\diamond$ III.\hspace{0.5mm}] Let us add two boundary prime ends  $b$ and $d$ so that $(a,b,c,d)$ are listed counter-clockwise.  Given a discrete Exploration Process, we may parametrize it in any convenient fashion and denote the resulting curve by $\mathbb X_t^\varepsilon$.  Let us assume, temporarily, that 
$\mathbb X_t^\varepsilon$ does not ``explore'' the 
boundary, $\partial \Omega_\varepsilon$.
Now, by convention/definition, the faces on the right side of the Exploration Process are blue, and the faces on the left side are yellow. In general, a blue crossing from $[a,b]$ to $[c,d]$ can either touch the blue portion of the exploration path $\mathbb X^{\varepsilon}_{[0, t]}$, or avoid it.  
It is thus a fact that the blue crossing in 
$\Omega_\e$ of the described type implies a blue crossing between $[\mathbb X^{\varepsilon}_t,b]$ to $[c,d]$ in $\Omega_\e\setminus\mathbb X^{\varepsilon}_{[0,t]}$.
And vice versa:  It is clear (at least modulo cases where $\mathbb X_{[0, t]}^\e$ touches $\partial \Omega_\e$) that any blue crossing between $[\mathbb X^{\varepsilon}_t,b]$ to $[c,d]$ in $\Omega_\e\setminus\mathbb X^{\varepsilon}_{[0,t]}$ produces a blue crossing from $[a,b]$ to $[c,d]$ in $\Omega_\e$.

Under these conditions, we can write the following \emph{Markov identity} for the crossing probabilities
\begin{equation}
\label{eq:tot}
C_{\varepsilon}\left(\Omega_\e\setminus\mathbb X^{\varepsilon}_{[0,t]},\mathbb X^{\varepsilon}_t,b,c,d\right)=
C_{\varepsilon}\left(\Omega_\e,a,b,c,d\ |\ \mathbb X^{\varepsilon}_{[0,t]}\right).
\end{equation}
and further,
\begin{equation}\label{eq:tot1}
\mathbb E_{\mu_\varepsilon}\left[C_{\varepsilon}\left(\Omega_\e\setminus\mathbb X^{\varepsilon}_{[0,t]},\mathbb X^{\varepsilon}_t,b,c,d\right)\right]=
C_{\varepsilon}\left(\Omega_\e,a,b,c,d \right).
\end{equation}
Now let $0 < s < t$, then given some $\mathbb X_{[0, s]}^\e$, the same reasoning as above applied to $\Omega_\e \setminus \mathbb X_{[0, s]}^\e$ and the conditional measure $\mu_\e\left(\cdot \mid \mathbb X_{[0, t]}^\e\right)$ gives the \emph{martingale equation}  
\begin{equation}\label{eq:tot2}
\mathbb E_{\mu_\varepsilon} \left[C_{\varepsilon}\left(\Omega_\e\setminus\mathbb X^{\varepsilon}_{[0,t]},\mathbb X^{\varepsilon}_t,b,c,d\right) \mid \mathbb X_{[0,s]}^\varepsilon \right]=
C_{\varepsilon}\left(\Omega_\e\setminus\mathbb X^{\varepsilon}_{[0,s]},\mathbb X^{\varepsilon}_s,b,c,d\right).
\end{equation}
We will later establish a continuum version of this equation (see Equation \eqref{martingale}).

\begin{remark}\label{fmartingale}
Here, let us focus briefly on circumstances where $\mathbb X_{[0, t]}^\e$ has touched $\partial \Omega_\varepsilon$ -- which turns out to be highly likely -- or has even ``already determined'' the crossing game in $\Omega_\e$ -- which must happen eventually.
In case of the former but not the latter, the above equations require no further discussion provided we interpret 
$\Omega_\e\setminus\mathbb X^{\varepsilon}_{[0,t]}$ as 
the connected component of $c$ in
$\Omega_\e\setminus\mathbb X^{\varepsilon}_{[0,t]}$ $=:$ \text{Comp}$_{\Omega\setminus \mathbb X_{[0, t]}^\e}(c)$.  As for the latter, it
is not difficult to see that this occurs precisely when either $b$ or $d$ fail to lie in the boundary of \text{Comp}$_{\Omega_\e\setminus \mathbb X_{[0, t]}^\e}(c)$.  As such, the notation $C_{\varepsilon}\left(\Omega_\e\setminus\mathbb X^{\varepsilon}_{[0,t]},\mathbb X^{\varepsilon}_t,b,c,d\right)$ can no longer be literally read as ``the crossing probability in said domain with these marked boundary points'' since as least one of the relevant points is not actually in the boundary of the relevant domain.  Notwithstanding, we can and will use the notation $C_{\varepsilon}\left(\Omega_\e\setminus\mathbb X^{\varepsilon}_{[0,t]},\mathbb X^{\varepsilon}_t,b,c,d\right)$
even when $b$ or $d$ is not in \text{Comp}$_{\Omega_\e\setminus \mathbb X_{[0, t]}^\e}(c)$ with the understanding that in this case the relevant crossing probability is given by
\[ \begin{cases}1;~~~\mbox{if $\mathbb X_t^\e$ has hit $[c, d]$ before $[b, c]$}\\
0 ;~~~\mbox{if $\mathbb X_t^\e$ has hit $[b, c]$ before $[c, d]$}.
\end{cases}
\]
We will continue with this convention in the $\e \rightarrow 0$ case.  
\end{remark}


It is noted that for $\varepsilon > 0$, we are dealing with a discrete system and the above holds regardless of the parameterization scheme (provided that no overcounting is engendered); however, some care will be needed as we take the continuum limit.  In particular, the above equation with all $\e$ removed does not really make sense unless all curves $\mathbb X_{[0, t]}$ are endowed with a ``common'' parameterization.  The natural choice is the L\"oewner parameterization, but this requires some argument since the relevant topology for convergence is in the sup--norm (or \textbf{Dist} norm).

\vspace{0.5cm}

$\diamond$ IV.\hspace{0.5mm}] Now we show that it is possible to re--parameterize by the L\"oewner parameterization.  What will suffice for us is a statement to the effect that every ``L\"owener parameterization neighborhood'' in the support of $\mu^\prime$ contains a \textbf{Dist}--neighborhood.  (By the former it is meant that if $\gamma$ and $\gamma^\prime$ are endowed with the L\"owener parameterization, then the distance between them is taken to be $d_\mathcal L(\gamma, \gamma^\prime) = \sup_t |\gamma(t) - \gamma^\prime(t)|$; thus the converse of the above claim is obvious.)  We remark that the statement is essentially deterministic; we put in the proviso that we are in the support of $\mu^\prime$ just to ensure that the curves can be L\"oewner parameterized in the first place. 



Hereafter we shall restrict attention to the portion of the curves which have not yet entered the $\Delta$ neighborhood of $c$.  Our first claim is that (for $\eta \ll \Delta$), in fact, these portions of all curves in the same $\eta$--{\bf{Dist}} neighborhood are in fact close in the \emph{L\"owner} parameterization.  Indeed, 
\begin{lemma}\label{sup_cap}
Consider curves $\gamma$ emanating from $a$ which stay outside of the $\Delta$ neighborhood of $c$.  If $\textbf{\emph{Dist}}(\gamma_1, \gamma_2) < \eta$, then 
\[ |\mbox{\emph{Cap}}_{\mathbb H}(\gamma_1) - \mbox{\emph{Cap}}_{\mathbb H}(\gamma_2)| < C(\Omega, \Delta) \eta^\alpha\]
for some $\alpha > 0$ and some $\Omega$ and $\Delta$ dependent constant $C(\Omega, \Delta)$.  Here \emph{Cap}$_\mathbb H(\cdot)$ denotes the half plane capacity.
\end{lemma}
\begin{proof}  On $\mathbb H$, if two (compact) sets $A_1$ and $A_2$ and $\sigma$ close (even) in the Hausdorff metric, then by 
for example the Beurling estimates (see e.g., Corollary 3.80 in \cite{lawler}) 
\begin{equation}\label{beurling} |\mbox{Cap}_\mathbb H (A_1) - \mbox{Cap}_\mathbb H (A_2)| \leq C \sqrt{\sigma}\cdot  \mbox{diam}(\mathcal N_\sigma(A_1))^{3/2},\end{equation}
where $\mathcal N_\sigma(A_1)$ denotes the Hausdorff--$\sigma$ neighborhood of $A_1$ and $C$ is some constant (the estimate is equally good if we replace $\mathcal N_\sigma(A_1)$ by $\mathcal N_\sigma(A_2)$).  In our case, we are only assuming $\sigma$--closeness in the original domain $\Omega$ and therefore one could \emph{a priori} be concerned about distortions near the boundary.  However, this can be rectified with the aid of some distortion theorems.  Let us decompose $\Omega = \mathcal N_\delta(\partial \Omega) \cup [\Omega \setminus \mathcal N_\delta(\partial \Omega)]$ 
and similarly given two curves $\gamma_1$ and $\gamma_2$, we will write e.g., 
$\gamma_1 = \hat \gamma_1 \cup \bar{\gamma}_1$, where $\hat{\gamma}_1 = \gamma_1 \cap [\Omega \setminus \mathcal N_\delta (\partial \Omega)]$ and $\bar{\gamma}_1 = \gamma_1 \cap \mathcal N_\delta(\partial \Omega)$.  

First by the Distortion Theorems (for a more detailed argument along these lines, see the proof of Lemma \ref{nodoublevisit}) we know that if $\varphi: \Omega \rightarrow \mathbb H$, then 
\[ \varphi(\mathcal N_\delta(\partial \Omega)) \subset \mathcal N_{C^\prime\sqrt{\delta}}(\partial \mathbb H)\]
for some ($\Omega$ dependent) constant $C^\prime$ and hence bounding the capacity via the area of the corresponding strip, we have 
\[ \mbox{Cap}_{\mathbb H}(\bar \gamma_1),\mbox{Cap}_{\mathbb H}(\bar \gamma_2) \lesssim D\sqrt \delta\]
where $D$ is the diameter of the image of the complement of the $\Delta$--neighborhood of $c$ under $\varphi$ and we use $\lesssim$ to denote implied universal/$\Omega$--dependent constants.  Next we note that by the subadditive property of capacities, it is clear that 
\[ |\mbox{Cap}_{\mathbb H}(\gamma_1) - \mbox{Cap}_{\mathbb H}(\gamma_2)| \leq \mbox{Cap}_{\mathbb H} (\bar \gamma_1) + \mbox{Cap}_{\mathbb H} (\bar \gamma_2) + |\mbox{Cap}_{\mathbb H}(\hat \gamma_1) - \mbox{Cap}_{\mathbb H}(\hat \gamma_2)|\]
so we now estimate $|\mbox{Cap}_\mathbb H(\hat \gamma_1) - \mbox{Cap}_\mathbb H(\hat \gamma_2)|$.  But first, by another distortion estimate (see e.g., Corollary 3.19 in \cite{lawler}) we have 
\[ |\varphi^\prime(z)| \lesssim 1/\sqrt \delta\]  
and hence $d(\varphi(z_1), \varphi(z_2)) \lesssim \frac{\eta}{\sqrt \delta}$ if $z_1, z_2 \in \Omega \setminus \mathcal N_{\delta}(\partial \Omega)$ with $d(z_1, z_2) < \eta$ and we conclude that 
\[ d_H(\hat \gamma_1, \hat \gamma_2) \lesssim \frac{\eta}{\sqrt \delta}\]
where $d_H$ denotes the Hausdorff distance, from which it follows by \eqref{beurling} that 
\[ |\mbox{Cap}_\mathbb H(\hat \gamma_1) - \mbox{Cap}_\mathbb H(\hat \gamma_2)| \lesssim \frac{\sqrt \eta}{\delta^{1/4}}.\]
Combining the above estimates, we see that with proper choice of $\delta$ (which vanishes with $\eta$), the difference in capacities indeed differs by a fractional power of $\eta$.
%
%
\end{proof}
We may thus safely replace all parameterizations by the L\"oewner parameterization:   
\begin{cor}
Let $\gamma$ be a L\"oewner curve emanating from $a$ and staying outside of the $\Delta$ neighborhood of $c$, and let $\mathcal L_\sigma (\gamma)$ denotes the $\sigma$ L\"oewner parameterization neighborhood of $\gamma$.  Then there exists $\eta = \eta(\sigma, \Delta, \gamma) > 0$ such that the the $\mathbf {Dist}$ neighborhood of size $\eta$ is contained in $\mathcal L_\sigma(\gamma)$.
\end{cor}
\begin{proof}
Suppose towards a contradiction that this is not the case.  Then there exists $\gamma_n \rightarrow \gamma$ in the $\mathbf{Dist}$ norm such that $d_\mathcal L(\gamma_n,    \gamma) > \sigma$.  It is clear that we may endow each $\gamma_n$ -- as well as $\gamma$ -- with some (uniform) parameterization so that $\sup_t |\gamma_n(t)  - \gamma(t)| = \eta_n$, which tends to zero; further, we can and will without loss of generality assume that $\gamma$ is in fact parameterized by capacity (this does \emph{not} imply that $\gamma_n$'s are parameterized by capacity; indeed, they are parameterized by $\gamma$'s capacity).  But this implies that there is a sequence of capacities $\mathfrak c_n$, which occur for $\gamma_n$ at time $s_n$ (in this parameterization) such that 
\[ |\gamma_n(s_n) - \gamma(\mathfrak c_n)| > \sigma. \]

Taking a subsequence if necessary, we may assume that $s_n \rightarrow s$.  Our first claim is that $\gamma_n([0, s_n])$ converges in the \textbf{Dist} norm to $\gamma([0, s])$.  Indeed, 
\[ \textbf{Dist}(\gamma_n([0, s_n]), \gamma([0, s])) \leq \textbf{Dist}(\gamma([0, s_n]), \gamma([0, s])) + \textbf{Dist}(\gamma([0, s_n]), \gamma_n([0, s_n])).\]
The second term is clearly bounded by $\eta_n$; as for the first term, it is clearly bounded by $\text{diam}(\gamma([s_n \wedge s, s_n \vee s]))$ which tends to zero since $\gamma$ is continuous.  We may assume without loss of generality (taking a subsequence if necessary) that $\mathfrak c_n \rightarrow \mathfrak c$.  
By Lemma \ref{sup_cap} we then have
$$
\mathfrak c = \lim_{n \to \infty}\text{Cap}_{\mathbb H}
(\gamma_n[0,s_n])  =  \text{Cap}_{\mathbb H}
(\gamma[0,s]).
$$
So using the fact that capacity is strictly increasing (which follows from the definition of L\"oewner curves) the above display implies
that $s$ = $\mathfrak c$ which is a contradiction since \textbf{Dist}--convergence necessitates that $\gamma_n(s_n) \to \gamma(s)$.
\end{proof}

\vspace{0.5cm}

$\diamond$ V.\hspace{0.5mm}]
As a first step towards obtaining a martingale observable in the continuum, our next goal is to remove all $\e$'s from \eqref{eq:tot1}.  On the basis of the previous step, it is clear that we may now interpret \eqref{eq:tot1} in terms of L\"oewner parameterization.  Further, we set $t > 0$ to be such that the relevant curves have not yet entered the $\Delta$ neighborhood of $c$.  First, the right hand side of \eqref{eq:tot} converges to the continuum counterpart $C_0(\Omega, a, b, c, d)$ by Lemma \ref{cardy_formula}, so we focus on the left hand side.

First, recalling that $\mu^\prime$ is a weak$^*$ limit with respect to the \textbf {Dist} norm, and that the space of all possible continuous curves is, in fact, separable, it follows that there are countably many curves $\gamma_n$ such that the space, $\mathscr C_{a, c, \Delta}$, of L\"oewner curves which begin at $a$ aiming towards $c$ but having not yet entered the $\Delta$ neighborhood of $c$, can be written as 
\[ \mathscr C_{a, c, \Delta} = \bigcup_{n=1}^\infty B_{\delta_n}(\gamma_n) \cap \mathcal L_{\sigma}(\gamma_n): = \bigcup_{n=1}^\infty \mathcal N_n^*.\]
In the above, $\delta_n$ has been chosen in accord with Lemma \ref{pt_eq_cont} (and also, for the model in \cite{cardy}, described in $\S$\ref{model}, Lemma \ref{admissibility} ensures that Cardy's Formula is viable for domains slit by the Explorer Process) so that $C_{\varepsilon}\left(\Omega_\e\setminus\mathbb X^{\varepsilon}_{[0,t]},\mathbb X^{\varepsilon}_t,b,c,d\right)$ for any $\mathbb X^{\varepsilon}_{[0,t]}$ in $B_{\delta_n}(\gamma_n)$ is $\vartheta$ close to the corresponding object with argument $\gamma_n([0, t])$ (for $\e < \e(\gamma_n)$ sufficiently small), where $\vartheta \ll 1$ is small, and $\sigma$ is also envisioned to be small.  Further, modifying the neighborhoods to be mutually disjoint, we can now reduce to a finite number, $N$, of these neighborhoods which carries all but $\alpha$ (with $\alpha \ll 1$) of the measure of $\mu^\prime$.  For what follows, we will sometimes abbreviate, e.g., 
\[K_\e (Y_t^\e) := C_{\varepsilon}\left(\Omega_\e\setminus\mathbb Y_{[0,t]}^\e,\mathbb Y_t^\e,b,c,d\right).\]
In the above display, it is understood that the right hand side is interpreted in accord with Remark \ref{fmartingale} above.

We first observe that (for $\e$ sufficiently small) 
\[ \left|\mathbb E_{\mu_\e}(K_\e(X_t^\e)) - \sum_{n=1}^N \mu_\e(\mathcal N_n^*) K_\e(\gamma_n)\right| \leq \alpha + \sum_{n=1}^N \sum_{X_t^\e \in \mathcal N_n^*}|K_\e(X_t^\e) - K_\e(\gamma_n)|~\mu_\e(X_t^\e) \leq \alpha + \vartheta\]
and similarly
\[ \left|\mathbb E_{\mu^\prime}(K_0(X_t)) - \sum_{n=1}^N \mu^\prime(\mathcal N_n^*) K_0(\gamma_n)\right| \leq \alpha + \sum_{n=1}^N \int_{X_t \in \mathcal N_n^*}|K_0(X_t) - K_0(\gamma_n)|~d\mu^\prime(X_t) \leq \alpha + \vartheta\]
Therefore, it is enough to control the difference of the relevant sums over neighborhoods:
\[ \begin{split}&~~~\left|\sum_{n=1}^N \mu_\e(\mathcal N_n^*) K_\e(\gamma_n) - \sum_{n=1}^N \mu^\prime(\mathcal N_n^*) K_0(\gamma_n)\right|\\
&\leq \sum_{n=1}^N \left|\mu_\e(\mathcal N_n^*) K_\e(\gamma_n) - \mu^\prime(\mathcal N_n^*) K_0(\gamma_n)\right|\\
&\leq \sum_{n=1}^N \left| \mu_\e(\mathcal N_n^*) ( K_\e(\gamma_n) - K_0(\gamma_n))| + |(\mu^\prime(\mathcal N_n^*) - \mu_\e(\mathcal N_n^*))K_0(\gamma_n)\right|\\
&\leq \vartheta + \alpha
\end{split}\]
Thus, taking $N \rightarrow \infty$ and $\e \rightarrow 0$, etc.,  we may now upgrade Eq.~\eqref{eq:tot} with
\begin{equation}\label{eq1}
\mathbb E_{\mu^{\prime}}\left[C_{0}\left(\Omega\setminus\mathbb X_{[0,t]},\mathbb X_t,b,c,d\right)\right]=
C_{0}\left(\Omega,a,b,c,d \right).
\end{equation}

\begin{remark}
The demonstration of Equation \eqref{eq1} (or some version thereof) \emph{in the continuum} represents the key issue in this approach to proving convergence.  In the present work, this has been achieved via a robust convergence to Cardy's Formula in general (i.e., slit) domains via the sup--approximations; see e.g., \cite{pt2}, Corollary 4.10.  In any case, the authors strongly believe that \emph{some} analytical statement along these lines cannot be avoided.  
\end{remark} 

Next we recast Equation \eqref{eq1} in terms of conditional expectation: 
\begin{equation}\label{cond_exp} \mathbb E_{\mu^\prime}(\mathbf 1_{\mathscr C_\Omega}\mid \sigma([0, t])) \equiv \mathbb E_{\mu^\prime}(\mathbf 1_{\mathscr C_\Omega}\mid \mathbb X_{[0, t]})=K_0(\mathbb X_{[0, t]}),\end{equation}
where $\sigma([0, t])$ denotes the $\sigma$--algebra generated by $\mu^\prime$ supported curves up to time $t$ and $\mathbf 1_{\mathscr C_\Omega}(\cdot)$ is the indicator function of the crossing event.  (The latter can be realized as 
\[ \mathbf 1_{\mathscr C_\Omega}(\gamma) = \begin{cases} 1 ~~~&\mbox{if $\gamma$ hits $[c, d]$ before $[b, c]$}\\
0 &\mbox{if $\gamma$ hits $[b, c]$ before $[c, d]$}
\end{cases}\]
and hence is a $\mu^\prime$ measurable function.)  
Note that e.g., $\mathbf 1_{\mathscr C_\Omega} \equiv 1$ if $\mathbb X_{[0, t]}$ has already hit the $[c, d]$ boundary of $\Omega$ and, in this vein, Equation \eqref{cond_exp} is of course interpreted in accord with Remark \ref{fmartingale} above.  We see that Equation \eqref{cond_exp} follows immediately: For $\mathcal B \in \sigma([0, t]])$, 
\[~~~\int_{\mathcal B} [\mathbb E_{\mu^\prime}(\mathbf 1_{\mathscr C_\Omega}\mid \sigma([0, t]))](\gamma)~d\mu^\prime(\gamma)
= \int_{\mathcal B} \mathbf 1_{\mathscr C_\Omega}(\gamma)~d\mu^\prime(\gamma)
= \mu^\prime(\mathscr C_\Omega \cap \mathcal B)
= \int_\mathcal B K_0(\mathbb X_{[0, t]})~d\mu^\prime.
\]
Here the first two equalities are definitions and the third equality can be established by a straightforward modification of the argument used to establish Equation \eqref{eq1} -- which corresponds to the case where $\mathcal B$ is the full sample space.  


From Equation \eqref{cond_exp} and the defining properties of conditional expectation, we can deduce that 1) the random variable $K_0(\mathbb X_{[0, t]})$ is $\sigma([0, t])$ measurable and 2) $K_0(\mathbb X_{[0, t]})$ is a continuous time martingale, i.e., if $0 < s < t$, then  
\begin{equation}\label{martingale}
\mathbb E_{\mu'}\left[C_{0}\left(\Omega \setminus \mathbb X_{[0, t]},\mathbb X_t,b,c,d\right)\ |\ \mathbb X_{[0,s]}\right] = C_0\left(\Omega \setminus \mathbb X_{[0, s]},\mathbb X_s,b,c,d\right).\end{equation}
In particular, Equation \eqref{martingale} is simply Equations \eqref{eq1} and \eqref{cond_exp} with $\Omega$ replaced by $\Omega \setminus \mathbb X_{[0, s]}$ -- along with the interpretation of the latter in terms of conditional expectations -- and $\mu^\prime$ averaging over $\mathbb X_{[s, t]}$.  More specifically, since $\sigma([0, s]) \subset \sigma([0, t])$, if $\mathcal B \in \sigma([0, s])$, then 
\[\begin{split} \int_\mathcal B \mathbb E_{\mu^\prime} (\mathbf 1_{\mathscr C_\Omega} \mid \sigma([0, s])~d\mu^\prime = \int_\mathcal B \mathbf 1_{\mathscr C_\Omega}~d\mu^\prime &= \int_\mathcal B \mathbb E_{\mu^\prime} (\mathbf 1_{\mathscr C_\Omega} \mid \sigma([0, t]))~d\mu^\prime\\
&= \int_\mathcal B \mathbb E_{\mu^\prime} \left[\mathbb E_{\mu^\prime} \left(\mathbf 1_{\mathscr C_\Omega} \mid \sigma([0, t]) \mid \sigma([0, s])\right)\right]~d\mu^\prime,\end{split}\]
which is the content of \eqref{martingale}.

\vspace{0.5cm}

$\diamond$ VI.\hspace{0.5mm}] We will now finish the proof and show that $\kappa = 6$.   Notice that the map 
$$h_t(z)=\frac{g_t(z)-g_t(d)}{g_t(b)-g_t(d)},$$
where $g_t(z)$ is the L\"owner map, maps the rectangle $(\Omega \setminus \mathbb X_{[0, t]}, \mathbb X_t,b,c,d)$ conformally onto 
\[\left(\mathbb H, \frac{\lambda_t-g_t(d)}{g_t(b)-g_t(d)},1,\infty,0\right).\]
By Cardy's identity (Lemma \ref{cardy_formula}), 
\begin{equation}\label{eq:cardy_cont}
C_0(\Omega \setminus \mathbb X_{[0, t]}, \mathbb X_t,b,c,d)=
F\left(\frac{g_t(b)-\lambda_t}{g_t(b)-g_t(d)}\right),
\end{equation}
where we recall that the relevant domain is really the connected component of $c$ in $\Omega \setminus \mathbb X_{[0, t]}$ and it is tacitly assumed that $b$ and $d$ are both (still) in the boundary of this component.

Using Eq.~\eqref{eq:cardy_cont}, we can rewrite Eq.~\eqref{martingale}, accounting for such errors, via
\begin{multline}\label{eq:total1}
\left|F\left(\frac{g_s(b)-\lambda_s}{g_s(b)-g_s(d)}\right)
\mathbf 1_{\{b, d \in \partial (\Omega \setminus \mathbb X_{[0, s]})  \}}-
\mathbb E_{\mu^\prime}\left(F\left(\frac{g_t(b)-\lambda_t}{g_t(b)-g_t(d)}\mid \mathbb X_{[0,s]}\right)\cap\{ b, d\in\partial(\Omega \setminus \mathbb X_{[0, t]})\}\right)\right|\\
\\
\leq\mathbb P\left(b\not\in\partial(\Omega \setminus \mathbb X_{[0, t]}) \text{ or }d\not\in\partial(\Omega \setminus \mathbb X_{[0, t]})\right).
\end{multline}

Let us now consider $|g_t(b)|$ and $|g_t(d)|$ both large compared with $\lambda_t$ and $t$, which may be enabled by considering $t$ fixed and $b, d \rightarrow c$.  In particular, let us define $b_0 = g_0(b)$ and $d_0 = g_0(d)$; the object $b_0$ will be our large parameter and since $b_0 > 0$ while $d_0 < 0$, we may as well defined $d_0$ via $d_0 = -r b_0$ with $r > 0$ of order unity.  It turns out that $r =1$ is slightly peculiar (which is anyway easily understood) and we will assume that this is \emph{not} the case.  

Let us observe right now that (for fixed $t$) the right hand side of Equation \eqref{eq:total1} tends to zero as we take $b_0$ to infinity: Since the capacity of the curve at time $t$ is, by definition, $2t$, it is clear that the image of the curve must stay within a distance $\approx \sqrt t$ of the real axis.  The asymptotic expansion for $g_t(g_0^{-1})$ directly implies that for for $t \ll b_0$, e.g., $b_t \approx b_0$ and hence, the image of the Exploration Process up to time $t$ under the map $g_0$ will be forced to cross a box of large aspect ratio, which by (conformal invariance of) Cardy's Formula, tends to zero exponentially like $e^{-\mathbf O(b_0/\sqrt{t})}$.  Hence it is sufficient to complete the proof under the assumption that $b, d \in \partial \Omega_t$.  

We now carry out the promised asymptotic expansion in $1/b_0$.  Recall that by the L\"owner parameterization in the half plane, $g_t(g^{-1}_0(z))=z+2t/z+\mathbf O(1/z^2)$ for $z\to\infty$.  Thus, 
\[ g_t(b) = b_0 + \frac{2t}{b_0} + \dots.\]
Therefore
(assuming $b$ and $d$ are in $\partial \Omega \setminus \mathbb X_{[0, s]}$) we may
write, for the first term on the left hand side of Eq.~\eqref{eq:total1} 
\begin{equation}\label{eq:taylor} F\left(\frac{g_t(b)-\lambda_t}{g_t(b)-g_t(d)}\right) = A(r) + B(r) \left[\frac{\lambda_t}{b_0}\right] + C(r) \left[\frac{\lambda_t^2 - 6t}{b_0^2}\right] + \mathbf O(b_0^{-3}).\end{equation}
We will need to take expectation of all terms; provided that each term in the expansion is well--defined, we may examine coefficients of various powers of $b_0$ and draw conclusions.  The necessary moment estimates appear in Lemma \ref{lem:apriori} below.  

First let us take expectations and note that Equation \eqref{eq1} implies that the average over $\mathbb X_{[0, t]}$ and hence $\lambda_{[0, t]}$ must provide the same result as in the original setup (corresponding to $t = 0$).  This implies, from the first two terms, that 
\begin{equation}\label{no_drift} \mathbb E(\lambda_t) = \lambda_0 = 0\end{equation}
and
\begin{equation}\label{2nd_moment} \mathbb E(\lambda_t^2 - 6t) = 0.\end{equation}
Finally, we reiterate that the entirety of $\mathbb X_{[0, t]}$ is determined by $\lambda_{[0, t]}$ (the history of the driving function up to time $t$).  Now, conditioning on $\mathbb X_{[0, s]}$ -- which is equivalent to conditioning on $\lambda_{[0, s]}$ -- Equation \eqref{martingale} gives us that the conditional expectation of Equation \eqref{eq:taylor} must (term by term) give us what we would have gotten with $s$ replacing $t$, namely, 
\[\mathbb E(\lambda_t\mid\lambda_s)=\lambda_s, \quad \mathbb E(\lambda_t^2-6t\mid \lambda_s)=\lambda_s^2-6s.
\]
Therefore both $\lambda_t$ and $\lambda_t^2-6t$ are continuous martingales, which, by L\'evy's characterization of Brownian Motion, implies that $\lambda_t$ has the law of $B_{6t}$.  Modulo the moment estimates for $\lambda_t$, this completes the proof of the Main Theorem. \qed

Finally, as promised, we will now prove an \emph{a priori} estimate on $\lambda_t$.

\begin{lemma}[\emph{A priori} Estimate]\label{lem:apriori}
$$\mathbb P[\lambda_t>n]\leq C_1\exp\left(-C_2\frac{n}{\sqrt{t}}\right),$$
for some absolute constants $C_1$ and $C_2$.
\end{lemma}

To prove Lemma \ref{lem:apriori} let us first observe that:

\begin{lemma}
Let $\gamma(t)$ be the chordal SLE generated by $\lambda_t$.  Then 
\begin{itemize}
\item $\mbox{\emph{Im}}(\gamma(t)) \leq 2\sqrt{t}$.
\item $\sup_{s \leq t} |\gamma(s)| \geq \frac{|\lambda_t|}{4}$.
\end{itemize}
\end{lemma}
\begin{proof}
We remark that the first statement (perhaps with a different constant) can be attained by capacity estimates, but in any case, let us observe that 
\[\partial_t (\mbox{Im}(g_t)) = -2\mbox{Im}(g_t)/|g_t - \lambda_t|^2 \geq -2/\mbox{Im}(g_t),\] 
so $\partial_t (\mbox{Im}(g_t))^2 /4 \geq -1$.  Integrating, we get
$ (\mbox{Im}(g_t))^2 \geq (\mbox{Im}(z))^2 - 4t$.  The conclusion is now clear if we plug in $z = \gamma(t)$ in the previous expression and note that $g_t(\gamma(t)) \in \mathbb R$.  

For the second part, let us denote $R_t = \sup_{s \leq t} |\gamma(s)|$.  From e.g., Corollary 3.44 of \cite{lawler}, we have that $|g_t(z) - z| \leq 3R_t$, for all $z \in \mathbb H \setminus \gamma([0, t])$.  The result follows by considering $z = \gamma(t)$ (or an approximating sequence).  
\end{proof}

Now we are in a position to prove Lemma \ref{lem:apriori}.  
\begin{proof}[Proof of Lemma \ref{lem:apriori}]
On the basis of the above lemma,  $|\lambda_t| > n$ implies that in the half plane a rectangle of aspect ratio of the order $n/\sqrt t$ has been crossed by $g_0(\gamma_{[0, t]})$.  But this means that $\gamma_{[0, t]}$ itself crossed a \emph{conformal} rectangle with conformal modulus $n/\sqrt t$.  Invoking Lemma \ref{cardy_formula}, the probability of such an event is bounded by $C_1 e^{-C_2 \frac{n}{\sqrt t}}$ for some $C_1, C_2 > 0$.   
\end{proof}

\section{Properties of Typical Explorer Paths}\label{proofs}
We will now provide proofs for the properties of a typical explorer path.  Recall that $\mu_\varepsilon$ is a measure generated by the percolation Exploration Process on the $\varepsilon$--lattice scale in a domain $\Omega$ with two distinguished boundary prime ends $a$ and $c$ and $\mu^\prime$ is any limit point of $\mu_\varepsilon$ in the weak$^*$--Hausdorff topology. 

\subsection{Estimates for Explorer Paths}\label{explorer_properties}
Here in this subsection, we collect some estimates for the explorer paths deduced from the underlying percolation systems.  
These estimates represent -- at the $\varepsilon$ level -- exactly the behavior that ensures that the limiting objects in the support of $\mu^\prime$ are precisely L\"owner curves.
We start with 
\begin{defn}\label{defdoubleback}
Let $\Omega$ be a domain.  Let $\delta \gg \eta > 0$ and let $\gamma: [0, 1] \rightarrow \Omega$ be a parametrized curve.  We say that $\gamma$ has a $\delta$--$\eta$ doubleback if there exists disjoint subsegments $I_1$ and $I_2$ of $[0, 1]$, with $\mbox{diam}(\gamma(I_1)) \geq \delta$, $\mbox{diam}(\gamma(I_2)) \geq \delta$, and such that the segments $\gamma(I_1)$ and $\gamma(I_2)$ are $\eta$--close in the sup--norm.
\end{defn}
\begin{lemma}[No Doubleback]\label{nodoubleback}
Let $\Omega$ be a domain and let $\gamma \in \mbox{supp}(\mu^\prime)$.  Let $\delta, \eta > 0$ satisfy $\eta < c_1 \delta$, with a particular $c_1$ of order unity.  Then for all $\delta$ sufficiently small, there are additional constants $c_2$ and $c_3$ of order unity such that for all $\varepsilon$ sufficiently small, the $\mu_\varepsilon$--probability of a $\delta$--$\eta$ doubleback is bounded above by 
\[ \frac{c_2}{\delta^2} \cdot e^{-c_3 \delta/\eta},\]
with the same result inherited by $\mu^\prime$.
\end{lemma}
\begin{proof}
It is sufficient to verify the statement in the measures $\mu_\varepsilon$ for $\varepsilon$ sufficiently small.  Thus let $\delta \ll 1$ and $\eta$ small as desired and then $\varepsilon$ much smaller than the scale set by $\eta$. 
Thus we are back to percolation estimates which reduce to crossing estimates for large boxes.  Proofs of similar results have appeared in the literature (many times) before so we shall be succinct.  In summary, the probability of a percolation path crossing a fixed box with aspect ratio of order 
$\delta:\eta$ is of order 
$e^{-[\text{const.}]\delta/\eta}$.  The event in question implies such a crossing (somewhere) and the factor of $\delta^{-2}$ accounts for all possible locations.  We now proceed.

For $k$ large but of order unity, let us grid the domain $\Omega$ into pixels of scale $k^{-1} \delta$.  It's not difficult to see that the event in question necessitates an easy--way $\eta$--close double--crossing of some rectangle of this scale with aspect ratio of order unity.  
Let us now consider a particular such $\delta: k\delta$ rectangle, denoted by $R_\delta$ and let us consider the event of at least two disjoint blue crossings of $R_\delta$ that are within distance $\eta$ of each other.  If $g_0$ is such a (single) crossing, let 
\[ N(g_0) = \{ \mbox{$\exists$ a blue crossing of $R_\delta$ in the region \emph{above} $g_0$ that is within distance $\eta$ of $g_0$}\}.\]
Our first claim is that, uniformly in $\varepsilon$, for all $\varepsilon$ sufficiently small, $\mathbb P(N(g_0)) \leq e^{-c_3\frac{\delta}{\eta}}$, for all $\eta, \delta$.  To see this, let us cover $g_0$ with disjoint annuli of scale $3\eta: \eta$, with the center of each annulus centered on a point of $g_0$.  Clearly, there are at least of the order $\delta/\eta$ such annuli.  If in the region above $g_0$, in any one of these annuli there is a yellow circuit, then $N(g_0)$ cannot possibly occur.  For future reference, we note that 
in fact these preventative steps take place in the intersection of the relevant annuli with $R_{\delta}$.
Since the probability of such a yellow circuit is uniformly positive, we have so far indeed shown that  
$$
\mathbb P(N(g_0)) \leq e^{-c_3 \frac{\delta}{\eta}}.
$$
Letting $\mathbf {G_0}$ denoting the event that $g_0$ \emph{is} the lowest crossing, one obtains the same estimate as the above for $\mathbb P(N(g_0) \mid \mathbf{G_0})$.  The estimates will hold if we now let $\mathbf {G_k}$ denote the event that the curve $g_k$ is the $k^{\mbox{th}}$ to lowest crossing, e.g., out of a total of $\ell \geq k$ disjoint crossings.  Thus, by subadditivity, conditioned on the existence of say $\ell$ disjoint crossings, the ultimate double--crossing event of interest has probability bounded above by $\ell e^{-d_3 \frac{\delta}{\eta}}$.  However, if $r_\ell$ denotes the probability of $\ell$ disjoint crossings in $R_\delta$, then by a BK--type inequality (which for the model at hand is provided in Lemma \ref{bklemma}) it is clear that $\sum_\ell \ell r_\ell < \infty$.  Hence the probability of two disjoint blue crossings (or two disjoint yellow crossings) in $R_\delta$ is bounded above by 
\begin{equation}\label{bfb}  c_2 e^{-c_3 \frac{\delta}{\eta}}.\end{equation}
To finish we note that there are only of order $\delta^{-2}$ such rectangles in $\Omega$ and hence summing over them, we have finished proving the lemma.  
\end{proof}
%

In the above and in what is to follow, results are shown to hold ``uniformly in $\varepsilon$ for $\varepsilon$ sufficiently small''
-- which, ultimately, always follows from scale invariance of the RSW estimates.  Hereafter we shall be somewhat less 
explicit concerning this matter. 

\begin{lemma}[Multi--Arm Estimates]\label{arms}
Let $D(\eta, l)$ denote the circular annulus with inner radius $\eta$ and outer radius $l$.  Consider the events of a (i) 5--arm crossing of $D(\eta, l)$ and (ii) 6--arm crossing of $D(\eta, l)$.  Then the 5--arm event has probability bounded above by $(\eta/l)^2$ while the 6--arm event has probability bounded above by $(\eta/l)^{2 + \sigma}$ for some $\sigma > 0$.  
\end{lemma}

\begin{proof}
Let us rescale back so that the lattice spacing is of order unity and the diameter of $\Omega_\varepsilon$ is of order $N$.  Then the five arm event in $D(\eta, l)$ is the event of five crossings between circles of radius $\eta N$ and $l N$.  Approximating by appropriate ``square'' annular regions, the arguments of \cite{kesten87} may be used in generic circumstances (of course some degree of reflection symmetry for the underlying lattice has to be employed and in addition it has been checked that the fencing/corridor arguments in \cite{kesten87} apply) and so the probability of the five arm event in $D(\eta, l)$ is bounded above by a constant times $(\eta/l)^2$.  
For the particular percolation model at hand, such issues were  
dispensed with in the proof of Lemma 7.3 in \cite{cardy}.  To bound the 6--arm event (also the subject of Lemma 7.3 in \cite{cardy} but not handled with ease) we note that if we let $A$ denote the event of one crossing in the annular region, then the probability of $A$ is bounded by $\left(\frac{\eta}{l}\right)^\sigma$, for some $\sigma > 0$, by standard Russo--Seymour--Welsh arguments.  
Then letting $B$ be the event of 5 crossings in the annular region and applying a BK--type inequality to $A \circ B$ (which for the model at hand is given as Lemma \ref{bklemma}) we obtain the desired result.
\end{proof}  

\begin{defn}\label{deftriplevisit}
Let $\Delta_2 > \Delta_1$ (with $\Delta_2 \gg \Delta_1$ envisioned) and let $\gamma: [0, 1] \rightarrow \Omega$ be a curve.  We say that $\gamma$ has a \emph{$\Delta_2$--$\Delta_1$ triple visit} if there are times $t_a < t_1 < t_b < t_2 < t_c < t_3 <  t_d$ such that $\gamma(t_1), \gamma(t_2)$ and $\gamma(t_3)$ all lie within a single $\Delta_1$--neighborhood while $\gamma(t_a), \dots, \gamma(t_d)$ each lie a distance at least $\Delta_2$ from some point in this neighborhood.  For an illustration see Figure \ref{triple_visit}.
\end{defn}

\begin{figure}[ht]
\centering
\vspace{-0.9cm}
\subfigure[Triple visit in the interior]{
\scalebox{0.35}{\input{triple_visit}}
\label{triple_visit}
}
\hspace{1cm}
\vspace{3mm}
\subfigure[Double visit near the boundary]{
\includegraphics[width=54mm]{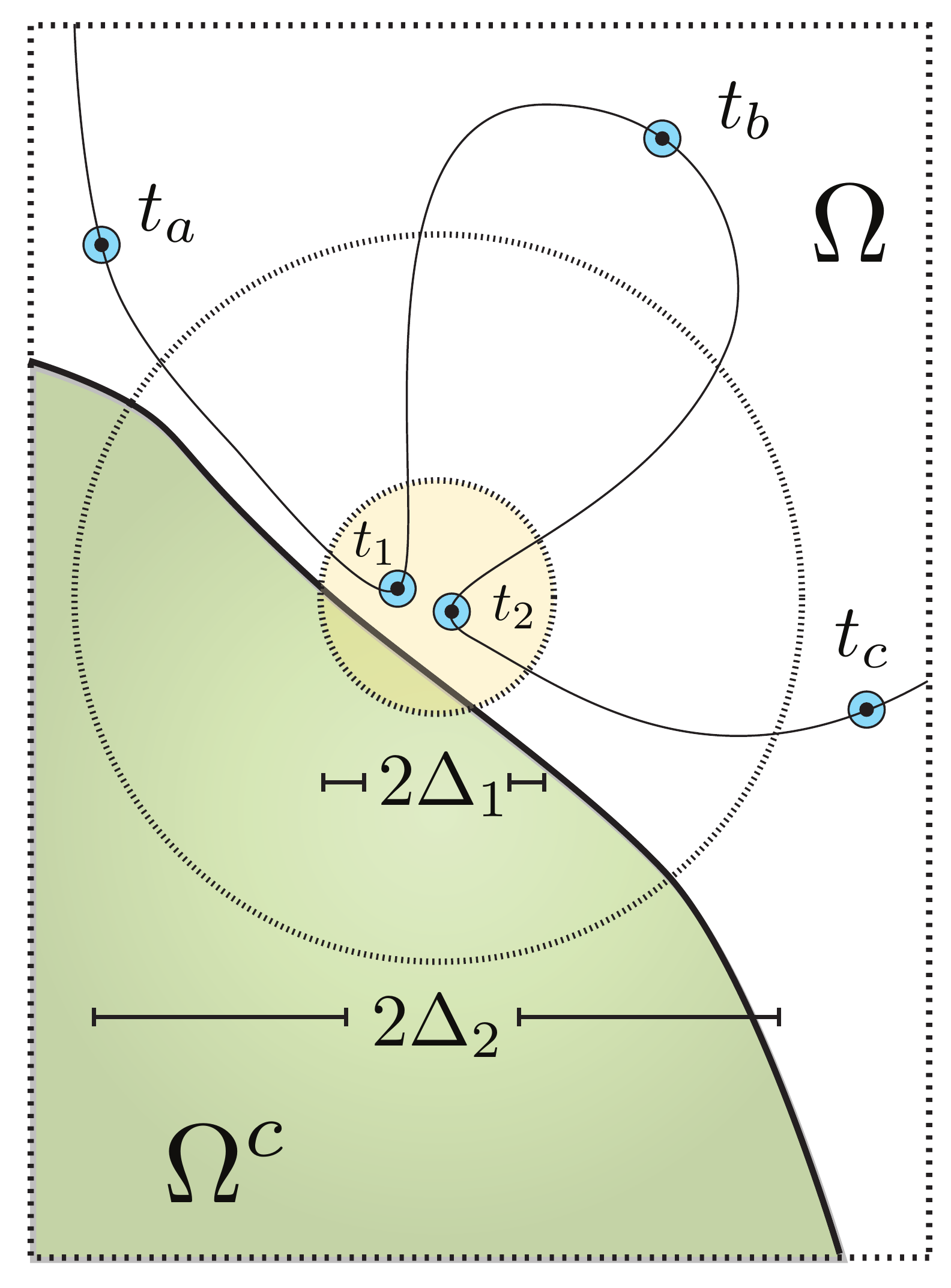}
\label{double_visit}
}
\caption[]{Atypical behavior of $\mu_\e$ curves}\label{triple_double} 
\end{figure}

A direct consequence of Lemma \ref{arms} is the absence of triple visits of the type described in the above definition as the ratio $\Delta_1/\Delta_2$ tends to zero:

\begin{lemma}\label{notriplevisit}
Let $\Omega$ be a domain and let $\Delta_2 \gg \Delta_1 > 0$.  The $\mu^\prime$--probability of a $\Delta_2$--$\Delta_1$ triple visit tends to zero as $\Delta_1/\Delta_2 \rightarrow 0$. 
\end{lemma}

\begin{proof}
A quick sketch of a triple visit scenario in $D(\eta, l)$ yields immediately 6 long disjoint passages of $\gamma(t)$ across the annulus.  Note this can occur in two topologically distinct fashions.  For $\gamma(t)$ a two--sided Exploration Process, na\"{\i}ve counting would yield as many as twelve long arms, but adjacent sides of ``disjoint'' long arms can lead to sharing of (boundary) elements of the process; in the worst possible case, entire adjacent arms can ``collapse''.  However, in either topology, even taking into account all these sharings and collapses, we are still left with six genuinely \emph{disjoint} long arms.  

We have established, in the continuum or lattice approximation, that the six arm event in an annulus $D(\eta, l)$ has probability bounded above by $\left(\frac{\eta}{l}\right)^{2+\sigma}$.  We may divide $\Omega$ (or $\Omega_\varepsilon$) into an overlapping grid of scale $\eta$.  The probability that such an event happens anywhere is therefore bounded above by $(\eta/l)^{2 + \sigma}\left(\frac{1}{\eta^2}\right)=\frac{1}{l^2}\left(\frac{\eta}{l}\right)^\sigma$, so ultimately, the probability of an actual triple visit is zero and the probability of a $\Delta_2$--$\Delta_1$ triple visit indeed tends to zero as $\frac{\Delta_1}{\Delta_2} \rightarrow 0$.
\end{proof}

\begin{remark}\label{nofuturevisit}
We make the following observation for intrinsic interest and for possible future reference: Observe that in one of the topological alternatives, after the second visit to the inner circle, the Exploration Process can immediately delve into the sack created between this visit and the first.  As an \emph{Exploration Process}, $\gamma(t)$ is now forced to perform its third visit \emph{and} escape $D(\eta, l)$ altogether.  The observation of interest is that these forced future visitation events provide, at least on the level of arm estimates, no additional decay after the (deep) visit into the cul--de--sac.  Indeed, six arms are already present at this juncture (all potential additional arms may undergo collapse).
\end{remark}

\begin{defn}\label{defnodoublevisit}
Let $\Omega$ be a domain.  Let $\Delta_2 > \Delta_1$ (with $\Delta_2 \gg \Delta_1$ envisioned) and let $\gamma: [0, 1] \rightarrow \Omega$ be a curve.  We say that $\gamma$ has a \emph{$\Delta_2$--$\Delta_1$ double visit to the boundary} by the obvious modification of Definition \ref{deftriplevisit} (using only $t_a, t_1, t_b, t_2, t_c$ along with the stipulation that at least one of the points $\gamma(t_1)$ or $\gamma(t_2)$ is within distance $\Delta_1$ of $\partial \Omega$).  For an illustration see Figure \ref{double_visit}.
\end{defn}

\begin{lemma}[No Double Visits Near the Boundary]\label{nodoublevisit}  
For any $\Delta_2 > 0$, the probability of a $\Delta_2$--$\Delta_1$ double visit to (anywhere on) the boundary tends to zero as $\Delta_1 \rightarrow 0$.

\end{lemma}

\begin{proof}
First we observe that if the Exploration Process has a $\Delta_2$--$\Delta_1$ double visit to the boundary, then this implies at least a 3--arm event on the scale of $\Delta_2: \Delta_1$ near the boundary.  This three--arm event can be viewed as the difference of crossing probabilities of certain conformal rectangles, all of which are contained in $\Omega$; the limiting probabilities of these events are therefore \emph{conformally invariant} and, furthermore, can be viewed under a single conformal map.  

The problem on the unit disc follows from well--known estimates: If $\mathcal N_{\mathbb D, p}$ denotes the $p$ neighborhood of the boundary in $\mathbb D$ then, as $\varepsilon \rightarrow 0$, the probability of a three--arm event between $\mathcal N_{\mathbb D, p_1}$ and $\mathcal N_{\mathbb D, p_2}^c$ is of the order $(p_1/p_2)^2$.  For percolation domains with smooth boundaries, this follows from the \emph{a priori} $1/N^2$ power law estimates described in \cite{aiz_chi} and \cite{lsw2}.  (The idea of proof is straightforward.  In brief: Consider the easy way crossing of an $N$ by $2kN$ box.  This probability is markedly larger than the similar probability in an $N$ by $kN$ box with both probabilities of order unity.  The difference between these two probabilities can be written as a telescoping sum, with each increment corresponding to a single site distortion, the vast majority of which leading to a three arm event in the half space -- the contributions from sites near the boundary are negligible.  This implies on the order of $N^2$ three arm events, each of which can be shown to happen with comparable probability by the rearrangement arguments of Kesten \cite{kesten87}.  Since the sum of all these probabilities is of order unity, the result follows). 

Let us then consider the uniformization map $\varphi: \mathbb D \rightarrow \Omega$.  We denote by $p_2 = p_2(\Delta_2, \Omega)$ the distance between $[\varphi^{-1}(\mathcal N_{\Omega, \Delta_2})]^c$ and $\partial \mathbb D$.  Obviously $p_2$ is independent of $\Delta_1$, therefore it is sufficient that the image of $\mathcal N_{\Omega, \Delta_1}$ is contained in a neighborhood of $\partial \mathbb D$ whose girth vanishes as $\Delta_1 \rightarrow 0$.  In particular and more than adequate it can be shown that $f(\mathcal N_{\Omega, \Delta_1}) \subset \mathcal N_{\mathbb D, C(\Omega) \sqrt{\Delta_1}}$: Indeed, by the Bieberbach Distortion theorem, $|\varphi'(z)|\geq |\varphi'(0)|(1-|z|)/4$.
By the Koebe 1/4-theorem,
$\mbox{dist}(\varphi(z),\partial\Omega)\geq1/4(1-|z|)|\varphi'(z)|\geq 1/16
|\varphi'(0)|(1-|z|)^2$. This implies the required estimate with
$C(\Omega)=4/\sqrt{|\varphi'(0)|}$.

\end{proof}

\begin{remark}\label{no_full_doublevisit}
The above estimates apply equally to the situation when the tip of the Exploration Process has ``just'' performed a double visit; i.e., the time $t_c$ in Definition \ref{nodoublevisit} is in fact superfluous.  This situation is analogous to the forced future triple visitations discussed in Remark \ref{nofuturevisit}.  As in these cases, the ostensible extra arms that the continuation of the journey might generate are susceptible to collapse and cannot be counted, while the estimates are already sufficient without these arms.  
\end{remark}

\subsection{Limit is Supported on L\"oewner Curves}\label{aizenman}

Here we provide a proof of Lemma \ref{tightness}, i.e., any limit point of the $\mu_\varepsilon$'s is supported on L\"oewner curves.  Our proof will utilize three additional lemmas, but first we must discuss crosscuts.  

As alluded to several times before, we envision $\Omega$ as the conformal image of the upper half plane via some map $\phi: \mathbb H \rightarrow \Omega$.  The prime end $a$ is defined in the usual fashion as the set of all limit points of sequences $\phi(z_n)$, $z_n \rightarrow z_a$, where $z_a \in \mathbb R$ is fixed.  Alternatively, consider 
\[ A_k = \overline{\phi(\{ |z - z_a| \leq 1/k, \mbox{Im}z > 0\})},\]
then the prime end $a$ can be defined as $\cap_k A_k$.  
We define similar quantities for $c$ and call them $C_k$.  Finally let us also define $\gamma^{k}_\varepsilon$ to be the curve formed by $\gamma_\varepsilon$ from the last exit from $A_k$ to the first entrance into $C_k$ after this last exit from $A_k$ (here $\gamma_\e$ denotes a generic $\mu_\e$ curve).  We remark that for finite $k$, with non--zero probability, $\gamma_\varepsilon$ \textit{will} form multiple crossings of the region $\Omega_k \equiv \Omega \setminus (A_k \cup C_k)$, but this probability tends to zero as $k \rightarrow \infty$, as can be seen by applying Cardy's Formula (or by using Russo--Seymour--Welsh type arguments, c.f.~the proof of Lemma \ref{admissibility}).

\begin{lemma}\label{curves}
Consider the domain $\Omega_k$ and let $\mu^\prime_k$ be a limit point of the measures on the curves $\gamma^{k}_\varepsilon$.  Then the $\mu^\prime_k$'s are supported on H\"older continuous curves.  Moreover, the weak convergence to $\mu^\prime_k$ can be taken with respect to the topology defined by the sup--norm distance between curves.
\end{lemma}  

\begin{proof}
These claims follow from the result of \cite{ab}.  
We claim that on $\Omega_k$, the curves $\{\gamma^{k}_ \varepsilon\}$ satisfy hypothesis H1 of \cite{ab}, namely:  
The probability of multiple crossings of circular shells (intersected with 
$\Omega_k$) goes to zero as the multiplicity gets large.  This is clear if  we consider circular shells with the outer radius sufficiently small, dependent on $k$.  Indeed, for $R$ less than some $R_k$, there is no possibility of both blue and yellow boundary inside $\Omega_k$ intersected with the corresponding circular shell.  Thus we must only rule out many crossings of $\gamma^{k}_\varepsilon$ of the circular shell either in the presence of no boundary or in the presence of a monochrome boundary -- with the rate of decay which increases to infinity with the number of traversals.  These estimates follow from straightforward repeated applications of the BK type inequality, which, for the model at hand, is proved in Lemma \ref{bklemma}.
\end{proof}

For the next lemma, we need another definition.  We say that we have a \emph{jump} of magnitude (at least) $\ell$ if 
\[ \gamma^{k+\ell}_\varepsilon \cap (\Omega_\varepsilon \setminus (A_k\cup C_k)) \neq \gamma_\varepsilon \cap (\Omega_\varepsilon \setminus (A_k \cup C_k)).\]
For an illustration see Figure \ref{D}.

\begin{figure}[htp]
\centering
\vspace{-5mm}
\includegraphics[width=0.80 \textwidth]{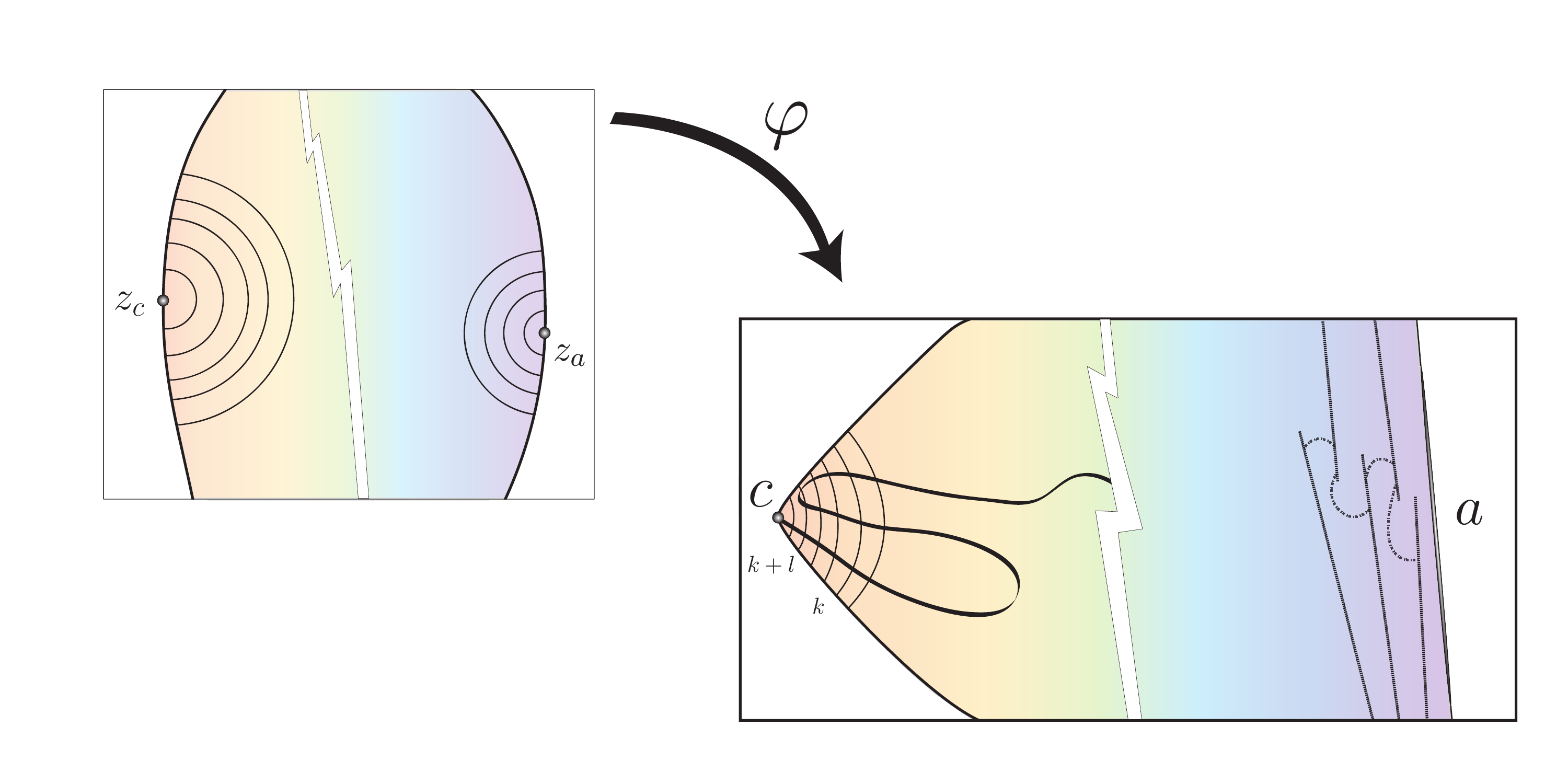}
\vspace{-5mm}
\caption{\footnotesize{
A jump of magnitude $l$ occurring in the vicinity of the prime end $c$.}}
\label{D}
\end{figure}

\begin{lemma}\label{jumps}
For every $k$, as $\varepsilon \rightarrow 0$, the magnitude of the jumps stay bounded with probability one.  
\end{lemma}
\noindent 
\begin{proof}
The modulus of the conformal rectangle $(A_k \setminus A_{k+\ell})^\circ$ tends to infinity as $\ell \rightarrow \infty$ (with $\ell \gg k$ envisioned).  We observe that in the event of a jump there must be a crossing of this conformal rectangle.  As $\varepsilon \rightarrow 0$, we may utilize Cardy's formula to show that the probability of such a crossing is bounded by some constant $\delta_{k, \ell}$ which tends to zero as $\ell \rightarrow \infty$, i.e., as $\varepsilon \rightarrow 0$, the probability of jumps of unbounded magnitude is zero.  Analogous arguments hold for the $C_k$'s.
\end{proof}    

We are now ready to prove that $\mu^\prime$ is supported on L\"oewner curves.

\begin{proof}[Proof of Lemma \ref{tightness}] 
We first establish that any limiting measure $\mu^\prime$ is supported on curves from $a$ to $c$.  By Lemma \ref{jumps}, a $\mu^\prime$ generic set intersected with $\Omega \setminus (A_k \cup C_k)$ is the same as $\mu_{k+\ell}^\prime$ generic curves (these objects are curves by Lemma \ref{curves}) intersected with $\Omega \setminus (A_k \cup C_k)$ for some $\ell$.  The family of domains $\Omega \setminus (A_k \cup C_k)$ is monotone and exhaustive, and hence $\mu^\prime$ is concentrated on curves.  By Lemma \ref{jumps} again, these curves are crosscuts from $a$ to $c$.  

To show that these are L\"owner crosscuts it is enough to show that they almost surely satisfy conditions (L1) and (L2).  Consider a parametrization of $\gamma$ with non--vanishing speed.  It is not difficult to see that a violation of (L1) implies that there exists some point $z_0$ which is visited at least three times if $z_0$ is in the bulk or twice if $z_0$ is on the boundary.  We remind the reader that this is in the continuum; at the lattice level, our collisions could represent approaches which are microscopically large but macroscopically small e.g., a sublinear power of $N$.  

Such an encounter in the interior leads to a triple visit and thus has vanishing probability, by Corollary \ref{notriplevisit}.   If $z_0$ is $\eta(\varepsilon)$--close to the boundary, $\eta \rightarrow 0$,  violation of (L1) implies a double visit below/at $z_0$.  As $\varepsilon \rightarrow 0$, this has vanishingly small probability, by Lemma \ref{nodoublevisit}.  Finally, a violation of (L2) is equivalent to the existence of some severe doubling back (e.g.~at scales $\delta(\varepsilon)$, $\eta(\varepsilon)$, with $\eta/\delta \rightarrow 0$), as defined in Definition \ref{defdoubleback} and therefore is forbidden by Lemma \ref{nodoubleback}.  
\end{proof}

We are now prepared to define the $\mathbf{Dist}$ function alluded to in the previous section.  
\begin{defn}\label{frechet}
Let $\lambda_\ell > 0$ be fixed numbers that satisfy $\sum_\ell \lambda_\ell = 1$, e.g., $\lambda_\ell = 2^{-\ell}$.  If $\gamma_r$ and $\gamma_g$ are two curves in $\Omega$ from $a$ to $c$, we denote, as before, $\gamma_r^\ell$ (or $\gamma_r^{\ell, \varepsilon}$) the appropriate portion of the curve in $\Omega_\ell$, etc.  Let $d_\ell(\gamma_r, \gamma_g)$ denote the usual sup norm distances between $\gamma_r^\ell$ and $\gamma_g^\ell$.  Then we define
\[ \mathbf{Dist}(\gamma_r, \gamma_g) = \sum_\ell \lambda_\ell d_\ell (\gamma_r, \gamma_g).\]
\end{defn}

As a corollary, we have weak$^*$ convergence of $\mu^\e$ to $\mu^\prime$ with respect to the topology provided by the \textbf{Dist} norm:

\begin{proof}[Proof of Lemma \ref{Dist_conv}]
For any finite $k$, we have by the result of \cite{ab} that $\mu_k^\prime$ is the weak$^*$ sup--norm limit of the objects $\mu_{k, \e}$, which are measures on the curves $\{\gamma^{k}_\e\}$.  It only remains to be seen that once two curves in $\Omega_k$ are close for $k$ large, then they remain close uniformly in $k$, but this is a property which follows directly from the definition of \textbf{Dist}.
\end{proof}

\subsection{Preservation of $M(\partial \Omega) < 2$}\label{preservation}

Here we show that if we start with some domain $\Omega$ with boundary Minkowski dimension less than two, then the Exploration Process also yields a curve with Minkowski dimension less than two. 

\medskip
\begin{proof} [Proof of Lemma \ref{admissibility}] Let $z \in \text{Int} (\Omega)$ and $g_\delta(z)$ the box of radius $\delta$ surrounding $z$ and $D(z)$ denote the distance between $z$ and $\partial \Omega$.  We claim that there is some $\psi > 0$ such that for all $\varepsilon$ sufficiently small, 
\[ \mathbb P_\varepsilon (\mathbb X_t^\varepsilon \in g_\delta(z)) < C_2 \left(\frac{\delta}{D}\right)^\psi\]
where $C_2$ is a constant.  

This follows from Russo--Seymour--Welsh theory, which we do here in some detail.  Indeed, if $r < s$, let $A_{s, r}(z) \equiv B_s(z) \setminus B_r(z)$ denote the annulus centered at $z$, where, if necessary, the sides are approximated, within $\varepsilon$, by the lattice structure.  Assume temporarily that $A_{s, r}(z) \subset \text{Int}(\Omega)$.  Clearly, if there is both a yellow and a blue ring in $A_{s, r}$, then $\mathbb X_t^\varepsilon$ cannot possibly visit $B_r(z)$ (since the yellow portion of $\mathbb X_t^\varepsilon$ cannot penetrate the blue ring and similarly with yellow $\leftrightarrow$ blue).  Now by the Russo--Seymour--Welsh estimates alluded to (Theorem 3.10, item (iii) in \cite{cardy} for the model at hand) the probability of a blue ring in $A_{M, \lambda M}$ is bounded below uniformly in $\varepsilon$ by a strictly positive constant that depends only on $\lambda$.  Let $\eta > 0$ denote a lower bound on the probability that in $A_{4L, 3L}$ there is a blue ring and in $A_{3L, 2L}$ a yellow.  Now let $k$ satisfy $2^k  > \varepsilon^{-1} D > 2^{k-1}$ and similarly $2^\ell > \varepsilon^{-1} \delta > 2^{\ell-1}$.  Then, give or take, there are $k-\ell$ independent annuli in which the pair of rings described can occur.  The probability that all such ring pair events fail is less than $C_1 (1 - \eta)^{k-\ell} \leq C_2 \left(\frac{\delta}{D}\right)^\psi$, where $C_1$ and $C_2$ are constants and $\psi > 0$ is defined via $\eta$.

Let us fix a square grid of scale $\delta$ with $\varepsilon << \delta << 1$.  Let $\mathcal N_\delta$ denote the number of boxes of scale $\delta$ that are visited by the process.  We claim that for all $\varepsilon$ sufficiently small
\begin{equation}\label{fore} \mathbb E_\varepsilon(\mathcal N_\delta) \leq C_{\psi^\prime} \left(\frac{1}{\delta}\right)^{2-\psi^\prime} \hspace{-3mm} = C_{\psi^\prime} n^{2-\psi^\prime},\end{equation}
where $\psi^\prime > 0$ is a constant and $n = n_{\delta} = \delta^{-1}$ represents the characteristic scale of $\Omega$ on the grid of size $\delta^{-1}$.  In particular we may take $\psi^\prime < \min\{\psi, \theta\}$, where $\theta \in [0,1]$ describes the roughness of the boundary: $M(\partial \Omega) = 2-\theta$.  

Let $n_k$ denote the number of boxes a distance $k\delta$ (i.e., $k$ boxes distant) from $\partial \Omega$ and 
\[N_l = \sum_{k \leq l} n_k.\]
Our first claim is that for all $\delta$,
\begin{equation}\label{boxes}N_l < C_{\theta^\prime} n^{2-\theta^\prime} l^{\theta^\prime}, \end{equation}
for any $\theta^\prime < \theta$, where $C_{\theta^\prime}$ is a constant.  To see this, let us estimate the total area of boxes on a grid of size $\sigma$ intersected by or within one unit of $\partial \Omega$.  It is not hard to see that this is bounded by $C_{\theta^\prime} \times \left(\frac{1}{\sigma}\right)^{2-\theta^\prime} \hspace{-2mm} \times \sigma^2 = C_{\theta^\prime}{\sigma^{\theta^\prime}}$, where $C_{\theta^\prime}$ is a constant which is uniform for a fixed $\theta^{\prime} < \theta$.  Taking $\sigma = l\delta$ and noting that \emph{these} boxes contain all of the $n_1 + \dots + n_l$ boxes of scale $\delta$ (i.e., boxes within $l$ units of $\partial \Omega$), the claim follows.

Now, clearly, 
\[ \mathbb E_\varepsilon (\mathcal N_\delta) \leq C_2 \sum_{k = 1}^{l_{\text{max}}} n_k \cdot \left(\frac{1}{k}\right)^\psi.\]
Let us now dispense with the sum in the display.  Summing by parts, we get
\[ \sum_{k = 1}^{l_{\text{max}}} n_k \left(\frac{1}{k}\right)^\psi = N_{l_{\text{max}}} l_{\text{max}}^{-\psi}  + \sum_{k = 1}^{l_{\text{max}}-1} N_k \left(\frac{1}{k^{\psi}}  - \frac{1}{(k+1)^\psi}\right).\]
Now if $\psi > \theta$, then $\psi > \theta^\prime$.  Using Eq.~(\ref{boxes}) and pulling out an $n^{2-\theta^\prime}$ , the sum is convergent.  Meanwhile, the first term (again using the estimate in Eq.~(\ref{boxes})) is smaller.  Conversely, if $\psi \leq \theta$, then both terms are of order $n^{2-\theta^\prime} l_{\text{max}}^{\theta^\prime - \psi}$ and the result follows if we take $l_{\text{max}} = n$.  It is re--emphasized that the estimate in Eq.~(\ref{fore}) is uniform in $\varepsilon$; by further sacrifice of the constant, we may claim that Eq.~(\ref{fore}) holds for all box--scales in the range $[\delta, 2\delta]$.
 
The remaining argument is now immediate.  Letting $\delta_k = 2^{-k}$ we have that for any $\delta \in [\delta_{k+1}, \delta_{k}]$ and $s >0$
\begin{equation}
\mathbb P_{\varepsilon}(\mathcal N_{\delta} 
> C_{\psi^{\prime}}n_{\delta}^{2-\psi^{\prime}+s})
\leq \frac{1}{2^{ks}}.
\end{equation}
The result follows, for any $s > 0$, by taking $\varepsilon \to 0$ and summing over $k$. 
\end{proof}

\section{The Model}\label{model_specific}

\subsection{Review of Model}
\label{model}
Here we give a quick description of the model under study.  For more details see Section 2.2 of \cite{cardy}.  The model takes place on the hexagon tiling of the 2D triangular site lattice: hexagons are yellow, blue and sometimes split; half and half.  Connectivity for us is defined by adjacent shapes (of the same color) sharing an edge segment in common.  Our description of the model starts with a particular local arrangement of hexagons: 

\begin{defn}
A \emph{flower} is the union of a particular hexagon with its six neighbors.  The central hexagon we call an \emph{iris} and the outer hexagons we call \emph{petals}.  We number the petals from 1 to 6, starting from the one directly to the right of the iris.  All hexagons which are not flowers will be referred to as \emph{filler}.
\end{defn}

\noindent Let $\Omega \subset \mathbb C$ be a domain, which for simplicity we may regard as being a finite connected subset of the hexagon lattice.  A \emph{floral arrangement}, symbolically denoted $\Omega_{\mathfrak F}$, is a designation of certain hexagons as irises (this determines the flowers).  There are three restrictions on placement of irises: (i) no iris is a boundary hexagon, (ii) there are at least two non--iris hexagons between each pair of irises, and (iii) ultimately in infinite volume the irises have a periodic structure with $60^\circ$ symmetries.

We are now ready to define the statistical properties of our model.

\begin{figure}[htp]
\centering
\vspace{-2mm}
\includegraphics[width=0.52 \textwidth]{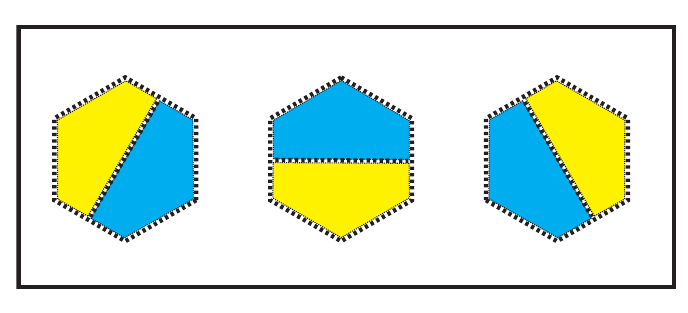}
\caption{\footnotesize{
The three allowed ``split'' states of the hexagon.  Note that these correspond to single bond occupancy events in the corresponding up--pointing triangle in the bond--triangular lattice percolation problem.}}
\label{A}
\end{figure}

\begin{defn}
Let $\Omega$ be a domain with floral arrangement $\Omega_\mathfrak F$.
\begin{itemize}
\item Petals and hexagons in the complement of flowers are only allowed to be blue or yellow, each with probability 1/2.
\item For ``most'' configurations of petals, irises can be blue, yellow, or mixed (one of three ways c.f., Figure \ref{A}) with probabilities $a$, $a$, or $s$, so that $2a + 3s = 1$ and in addition, 
\[ a^2 \geq 2s^2.\]  
\item The exceptional configurations of petals, which we call \emph{triggers}, are configurations where there are three yellow and three blue petals, with one pair of blue (and hence also yellow) petals contiguous.  In these configurations, the irises can now only be blue or yellow, each with probability 1/2.  
\end{itemize}  

Note that triggering is the only source of (very short range) correlation in this model; everything else is configured independently.  It is worth noting that for each floral arrangement, we have a one--parameter family of critical models with $s=0$ reducing to the usual site percolation on the triangular lattice.
\end{defn} 

Finally, it is remarked that the total of five possible configurations on a hexagon correspond to the eight possible configurations on (up--pointing) triangles -- of which there are five distinct connectivity classes.  It is not hard to see, by checking local connectivity properties, that the model described is a representation of a correlated percolation model on the triangular bond lattice.  

It was shown in \cite{cardy} Theorem 3.10 that our model exhibits all the typical properties of a 2D percolation model at criticality.  Cardy's formula for this model was the main result of \cite{cardy} (Theorem 2.4).  More specifically, let $\Omega \subset \mathbb C$ be a domain with piecewise smooth boundary which is conformally equivalent to a triangle.  Let us denote the three boundaries and ``prime ends'' of interest by $\mathcal A, c, \mathcal B, a, \mathcal C, b$, in counterclockwise order.  We endow $\Omega$ with an approximate discretization (with hexagons) on a  lattice of scale $\varepsilon = 1/N$ and a floral arrangement $\Omega_{\mathfrak F_\varepsilon}$.  Let $z$ be the vertex of a hexagon in $\Omega_{\mathfrak F_\varepsilon}$.  We define the discrete crossing probability function $u_\varepsilon^Y(z)$ to be the indicator function of the event that there is a blue path connecting $\mathcal A$ and $\mathcal B$, separating $z$ from $\mathcal C$, with similar definitions for $v_\varepsilon^Y(z)$ and $w_\varepsilon^Y(z)$ and the blue versions of these functions.  Then taking the scaling limit in an appropriate fashion (for more details see Section 2.3 of \cite{cardy}), we have, e.g.
\[ \lim_{\varepsilon \rightarrow 0} u_\varepsilon^Y = u,\]
where $u$ is one of the so--called Carleson--Cardy function: It is harmonic, and on the up--pointing equilateral triangle with base $\mathcal C$ being the unit interval, it is equal to $\frac{2}{\sqrt{3}}\cdot y$ -- this is equivalent to Cardy's formula.  The functions $v$ and $w$ are defined similarly.   

\subsection{The Exploration Process}\label{explorer}
We now give a (microscopic) definition of the percolation Exploration Process tailored to our system at hand.  We must start with a precise prescription of how to construct our domains.  Let $\Omega$ be a domain as described.  Let $a$ and $c$ be two prime ends and consider hexagons of the $\varepsilon$--tiling of $\mathbb C$.  It is assumed that within this tiling (with fixed origin of coordinates) the locations of all irises/flowers/fillers are predetermined.  We define $\Omega_\varepsilon$ to be the union of all fillers and flowers whose closure lies in the interior of $\Omega$.  It is assumed that $\varepsilon$ is small enough that both $a$ and $c$ are in the same lattice connected component of the tiling.  Other components, if any, will not be discarded but will only play a peripheral r\^ole.  With the exception of flowers, the boundary of the domain will be taken as the usual internal lattice boundary, which consists of the points of the set which have neighbors not belonging to the set.  If the lattice boundary cuts through a flower, then the whole flower is included as part of the boundary.  The notation for this lattice boundary will be $\partial_\varepsilon \Omega_\varepsilon$.  

Consider points $a_\varepsilon$, $c_\varepsilon$ which are on $\partial_\varepsilon \Omega_\varepsilon$ and are vertices of hexagons.  We call $(\Omega_\varepsilon, \partial\Omega_\varepsilon, a_\varepsilon, c_\varepsilon)$ \emph{admissible} if 
\begin{itemize}
\item $\Omega_\varepsilon$ contains no partial flowers.
\item $\partial_\varepsilon \Omega_\varepsilon$ can be decomposed into two lattice connected sets consisting of hexagons and/or halves of boundary irises, one of which is colored blue and one of which is colored yellow, such that $a_\varepsilon$ and $c_\varepsilon$ lie at the points where the two sets join and such that the blue and yellow paths are valid paths following the connectivity and statistical rules of our model; in particular, the coloring of these paths do not lead to flower configurations that have probability zero. 
\item $a_\varepsilon$ and $c_\varepsilon$ lie at the vertices of hexagons, such that of the three hexagons sharing the vertex, one of them is blue, one of them is yellow, and the third is in the interior of the domain. (See Figure \ref{B}.)
\end{itemize}

\begin{figure}[htp]
\centering
\vspace{-2mm}
\includegraphics[width=0.52 \textwidth]{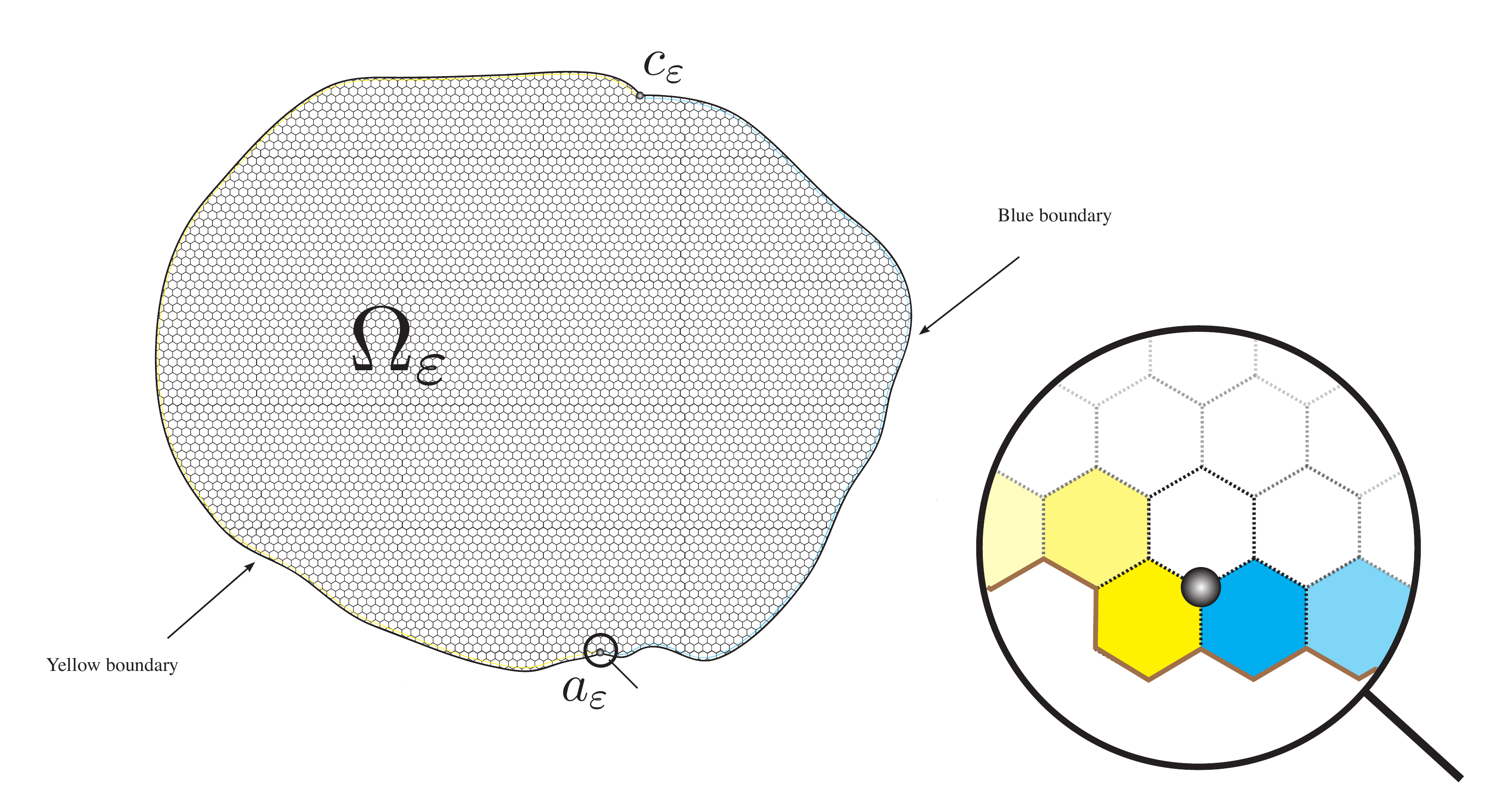}
\caption{\footnotesize{
The setup for the definition of the Exploration Process.}}
\label{B}
\end{figure}

We remark that in the case of boundary flowers (and other sorts of clusters on the boundary) it is not necessary to color \textit{all} the hexagons/irises.  Indeed the coloring scheme need not be unique -- it is only required that a boundary coloring of the requisite type can be selected.

It is not hard to see that the domains 
$(\Omega_\varepsilon, \partial_\varepsilon \Omega_{\varepsilon}, a_\varepsilon, c_\varepsilon)$ converges to $(\Omega, \partial \Omega, a, c)$ in the sense that $\partial_\varepsilon \Omega_\varepsilon$ and $\Omega_\varepsilon$ converge respectively to $\partial \Omega$ and $\Omega$ in the Hausdorff metric and in the Caratheodory metric with respect to any point inside $\Omega$.  Also, there exists $a_\varepsilon$ and $c_\varepsilon$ which converge respectively to $a$ and $c$ as $\varepsilon \rightarrow 0$.  Notice that the latter convergence is really in terms of the preimages under the uniformization map of the relevant domain.  In some sense we have chosen the ``simplest'' discretization scheme, which, in the companion work \cite{pt2} will be called the \emph{canonical} approximation; of course other discretizations are possible, but in the interest of brevity we shall not discuss these in the present work.

Geometrically, the \emph{Exploration Process} produces, in any percolation configuration on $\Omega_{\varepsilon}$, the unique interface connecting $a_\varepsilon$ to $c_\varepsilon$, i.e.,the curve separating the blue lattice connected cluster of the boundary from that of the yellow.  We denote this interface by $\gamma_\varepsilon$.  Dynamically, the exploration \emph{process} is defined as follows: Let $\mathbb X_0^\varepsilon = a_\varepsilon$.  Given $\mathbb X_{t-1}^\varepsilon$, it may be necessary to color new hexagons in order to determine the next step of the process.  (In particular, $\mathbb X_{t-1}^\varepsilon$ is ``usually'' at the vertex of a hexagon which has not yet been colored.)  We color any necessary undetermined hexagons according to the following rules:
\begin{itemize}
\item If the undetermined hexagon is a filler hexagon, we color it blue or yellow with probability 1/2.
\item If the undetermined hexagon is a petal or an iris, we color it blue or yellow or mixed with the conditional distribution given by the hexagons of the flower which are already determined. 
\item  If a further (petal) hexagon is needed, it is colored according to the conditional distribution given by the iris and the other hexagons of the flower which have already been determined.
\end{itemize}

\noindent We are now ready to describe how to determine $\mathbb X_{t}^\varepsilon$:
\begin{itemize}
\item If $\mathbb X_{t-1}^\varepsilon$ is not adjacent to an iris, $\mathbb X_t^\varepsilon$ will be equal to the next hexagon vertex we can get to in such a way that blue is always on the right of the segment 
$[\mathbb X_{t-1}^\varepsilon, \mathbb X_{t}^\varepsilon]$.
\item If $\mathbb X_{t-1}^\varepsilon$ is adjacent to an iris, then the state of the iris is determined as described above, after which the exploration path can be continued (keeping blue on the right) until a petal is hit.  The color of the petal will now be determined (according to the proper conditional distribution) and $\mathbb X_{t}^\varepsilon$ will equal one of the two possible vertices common to the iris and the new petal which keeps the blue region to the right of the final portion of the segments joining $\mathbb X_{t-1}^\varepsilon$ to $\mathbb X_{t}^\varepsilon$.
\end{itemize}
In particular, it is noted that at the end of each step, we always wind up on the vertex of a hexagon (see Figure \ref{thru}).  We denote by $(\gamma_\e)_t$ the actual value taken by the random variable $\mathbb X^\varepsilon_t$.   

\begin{figure}
\centering
\vspace{-2mm}
\includegraphics[width=0.65 \textwidth]{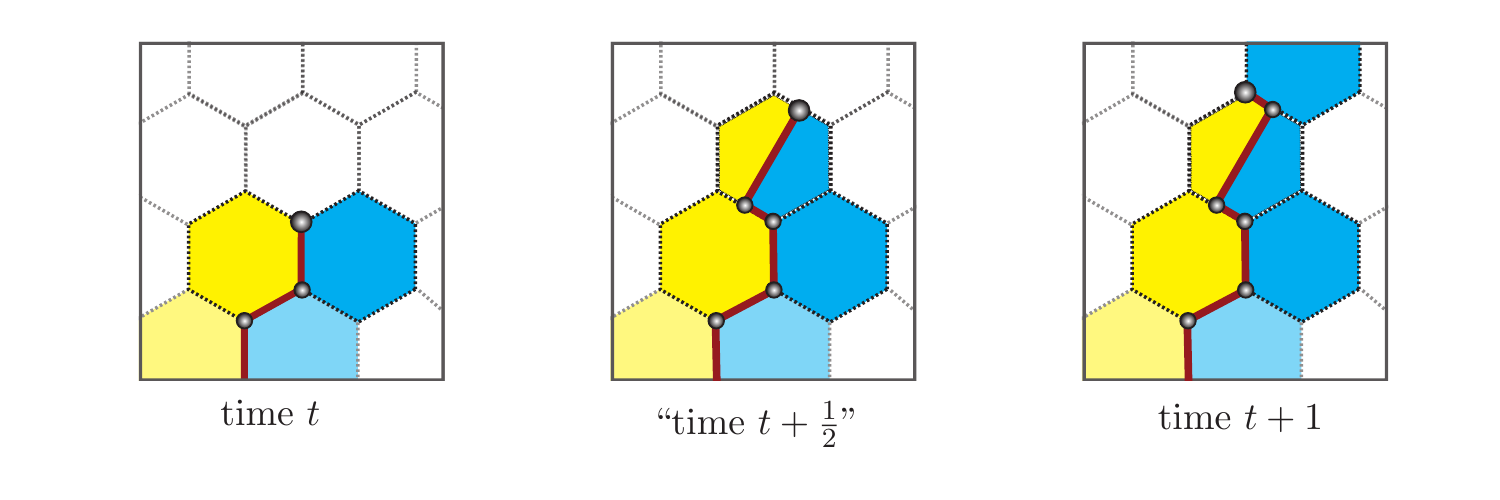}
\vspace{-5mm}
\caption{\footnotesize{
``Multistep'' procedure by which the Exploration Process gets through a mixed hexagon.}}
\label{thru}
\end{figure}

We state without proof some properties of our Exploration Process.

\begin{prop}
Let $\gamma_\varepsilon([0, t])$ be the line segments formed by the process up till time $t$, and $\Gamma_\varepsilon([0, t])$ the hexagons revealed by the Exploration Process.    
Let $\partial_\varepsilon \Omega_\varepsilon^t = \partial_\varepsilon \Omega_\varepsilon \cup \Gamma_\varepsilon([0, t])$ and let $\Omega_\varepsilon^t = \Omega_\varepsilon \setminus \Gamma_\varepsilon([0, t])$.  
Then, the quadruple  $(\Omega_{\varepsilon}^{t}, \partial_\varepsilon \Omega_\varepsilon^t, \mathbb X_{t}^\varepsilon, c_\varepsilon)$ is admissible.
Furthermore, the Exploration Process in $\Omega_\varepsilon^t$ from $\mathbb X_{t}^\varepsilon$ to 
$c_\varepsilon$ has the same law as the original Exploration Process from
 $a_\varepsilon$ to $c_\varepsilon$ in $\Omega_{\varepsilon}$ conditioned on $\Gamma_\varepsilon([0, t])$.  
 \end{prop}
 
 \subsection{A Restricted BK--Inequality}
Here we will prove an inequality that will be needed for proofs in several other places.

Suppose $A$ and $B$ are two events.  Then the BK inequality \cite{bk} states that (for suitable probability spaces) the probability of the \emph{disjoint} occurrence of $A$ and $B$ is bounded above by the product of their probabilities.  The most general version of this is Reimer's inequality \cite{reimer} (see also \cite{bcr} for more background and a self--contained proof), which holds for arbitrary product probability spaces.  For the model at hand, we do not have a product probability space; Reimer's inequality would, in the present context, yield the desired result only for \emph{flower} disjoint events.  Unfortunately, we have need of a stronger statement; specifically, for disjoint path--type events where the individual paths may use the same flower.  In fact, as the following example demonstrates, a general BK inequality does not hold in our system.  However, as we later show, an abridged version holds for path--type events.

\begin{ex}
Let $A$ be the event of a blue connection between petals 1, 4, and 5 (without any requirement on the color of the petals 1, 4, and 5), and let $B = \{\mbox{petals 1, 4, 5 are blue}\}$.  Observe that $B$ and $B^c$ are defined entirely on the petals 1, 4, 5, whereas $A$ is defined on the complementary set.  Therefore we have $A \cap B^c = A \circ B^c$.  By Example 6.1 of \cite{cardy}, we know that $\mathbb P(A \cap B) < \mathbb P(A) \mathbb P(B)$.  But this immediately implies that $\mathbb P(A \circ B^c) > \mathbb P(A) \mathbb P(B^c)$. 
\end{ex}

Before tending to the detailed analysis of flowers, let us first introduce the notion of disjoint occurrence for non--negative random variables.

\begin{defn}
Let $a_i, b_j \geq 0$ and let
\[ X = \sum_1^n a_i \mathbf 1_{A_i}, \hspace{3mm} Y = \sum_1^m b_j
\mathbf 1_{B_j},\] 
where $A_i \cap A_k = \emptyset$ for $i \neq k$ and $B_j \cap B_l = \emptyset$ for $j \neq l$.
We define
\[ X \circ Y = \sum_{i, j} a_i b_j \mathbf 1_{A_i \circ B_j}.\]
If the usual BK inequality holds then linearity immediately gives
\[ \mathbb E(X \circ Y) \leq \mathbb E(X) \mathbb E(Y).\]
\end{defn}
We will be working with this slight generalization; what we have in mind is the hexagon disjoint occurrence of paths, and in the
case of paths of different colors, sharing of the iris may occur.  To be precise, we have the following definition:

\begin{defn}\label{disjoint}
Let $\Omega_\mathfrak{F}$ denote a flower arrangement and let $S$ and $T$ denote sets in $\Omega_\mathfrak{F}$ which contain no irises.  Let $X^b_{S, T}$ denote the indicator of the event that all hexagons in $S$ and $T$ are blue and that there is a blue path -- possibly including irises -- connecting $S$ and $T$.  Similarly we define $X^y_{S, T}$ to be the yellow version of this event.  Now if $S^\prime$ and $T^\prime$ are two other sets of $\Omega_\mathfrak{F}$ which are disjoint from $S$ and $T$ and also do not contain irises, then we may define $X^b_{S, T} \circ X^b_{S^\prime, T^\prime}$ in accord with the usual fashion.  However, for present purposes, in the event corresponding to $X^b_{S, T} \circ X^y_{S^\prime, T^\prime}$, the two paths may share a mixed iris.   
\end{defn}

\begin{lemma}\label{bklemma}
Let $X_{S_1, T_1}^{\ell_1}, X_{S_2, T_2}^{\ell_2}, \dots, X_{S_n, T_n}^{\ell_n}$ be the indicator functions of path--type events as described in Definition \ref{disjoint}, where $\ell_i \in \{b, y\}$, then 
\[ \mathbb E(X_{S_1, T_1}^{\ell_1} \circ X_{S_2, T_2}^{\ell_2} \circ \dots \circ X_{S_n, T_n}^{\ell_n}) \leq \mathbb E(X_{S_1, T_1}^{\ell_1}) \mathbb E(X_{S_2, T_2}^{\ell_2}) \dots \mathbb E(X_{S_n, T_n}^{\ell_n}).\]
\end{lemma}

\noindent
\begin{proof} 
Our proof is slightly reminiscent of the proof of Lemma 6.2 in \cite{cardy}.  Let $\sigma$ denote a configuration of petals and filler and let $I$ denote a configuration of irises.  We will use induction; first we prove the statement for the case of exactly one flower (i.e., supposing there is only one flower in all of $\Omega_\mathfrak F$) and two path events,  whose indicator functions we denote by $X$ and $Y$. We write
\[ \mathbb E(X \circ Y) = \mathbb E_\sigma [\mathbb E_I (X \circ Y
| \sigma)].\] If we can show that $\mathbb E_I (X \circ Y |
\sigma) \leq \mathbb E_I(X | \sigma) \circ \mathbb E_I(Y |
\sigma)$, then we may apply the BK--inequality to the outer expectation to yield the desired result since, on the outside, the measure is independent.  It is clear that the function $\mathbb E(X \circ Y | \sigma)$ can only take on five different values; we write
\begin{equation}\label{XinY} \begin{split} \mathbb E(X \circ Y | \sigma) &= 1
\cdot \mathbf 1_{\mathcal O(X \circ Y)}(\sigma) \\
&+ (a + s) \cdot \mathbf 1_{A_1(X \circ Y)}(\sigma)\\
&+ (1/2) \cdot \mathbf 1_{A_2(X \circ Y)}(\sigma)\\
&+  (a + 2s) \cdot \mathbf 1_{A_3(X \circ
Y)}(\sigma)\\
&+ s \cdot \mathbf 1_{\mathcal F(X \circ
Y)}(\sigma),\end{split}\end{equation} where e.g.
\[\mathcal O(X \circ Y) = \{\sigma \mid \mathbb E(X \circ Y |
\sigma) = 1\}.\]  

It is not difficult to see that $\mathcal O(X \circ Y)$ is the set of $\sigma$ configurations where $X \circ Y$ has occurred on the complement of the iris.  The remaining terms warrant some discussion.  We first point out that these terms correspond to configurations where the flower is pivotal for the achievement of at least one of $X$ and $Y$, and, due to the nature of the events in question, petal arrangements in these configurations satisfy certain constraints.  For instance, configurations in $A_3$ must exhibit a petal arrangement such that one of the paths is in a position where it must transmit through the iris, which can be accomplished by the preferred color or two of the split configurations; the flower must \emph{not} be in a triggering configuration and, needless to say, the other path has already occurred (independent of the iris).  

Finally we observe that $\sigma \in \mathcal F(X \circ Y)$ implies that both paths must use the iris and therefore can only occur when the paths in question have different colors.  It is not hard to see, via petal counting, that $\mathcal F(X \circ Y)$ forces the alternating configuration of petals and that indeed, we have a situation of a ``parallel transmission'' through the iris, with exactly
one iris configuration which achieves both desired transmissions. We
also note that in similar expressions for $\mathbb E(X | \sigma)$
and $\mathbb E(Y | \sigma)$, the corresponding terms $\mathcal
F(X)$ and $\mathcal F(Y)$ will be empty, since e.g., if the path is blue and some iris is capable of achieving the transmission, then certainly the pure blue iris will achieve the transmission.

Let us expand $\mathbb E(X| \sigma) \circ \mathbb E(Y | \sigma)$ in the sense defined above:
\begin{equation}\label{XoutY} \begin{split} \mathbb E(X |
\sigma) \circ \mathbb E(Y | \sigma) &= 1 \cdot \mathbf
1_{\mathcal O(X) \circ \mathcal O(Y)}(\sigma) \\&+ (a + s)\cdot
[\mathbf 1_{\mathcal O(X) \circ A_1(Y)}(\sigma) + \mathbf 1_{A_1(X)
\circ \mathcal O(Y)}(\sigma)]\\&+ (1/2)\cdot [\mathbf 1_{\mathcal
O(X) \circ A_2(Y)}(\sigma) + \mathbf 1_{A_2(X)\circ \mathcal
O(Y)}(\sigma)]\\&+ (a + 2s) \cdot [\mathbf 1_{\mathcal O(X) \circ
A_3(Y)}(\sigma) + \mathbf 1_{A_3(X) \circ \mathcal
O(Y)}(\sigma)]\\&+ (a+s)^2 \cdot [\mathbf 1_{A_1(X) \circ
A_1(Y)}(\sigma)]\\&+ \mathcal R(a, s, \sigma),\end{split}\end{equation} where
$\mathcal R(a, s, \sigma)$ contains all the remaining terms in the expansion,
e.g. the terms
\begin{equation}\label{null}(1/2)(a+s)\cdot[\mathbf 1_{A_1(X) \circ
A_2(Y)}(\sigma) + \mathbf 1_{A_2(X) \circ
A_1(Y)}(\sigma)]\end{equation} and
\begin{equation}\label{useless}(a+s)(a+2s) \cdot [\mathbf 1_{A_1(X) \circ
A_3(Y)}(\sigma) + \mathbf 1_{A_3(X) \circ
A_1(Y)}(\sigma)].\end{equation} We claim that Eq.(\ref{null}) will
evaluate to zero for each $\sigma$: In the first term, $A_1(X)$ requires
that the petals exhibit a configuration which precludes a trigger
and $A_2(Y)$ requires the petals to exhibit a configuration which
leads to a trigger, and similarly for the second term.  The terms in Eq.(\ref{useless}) may or may not evaluate to zero for
all $\sigma$ \emph{a priori}, but in any case will not be needed.

Now we match up the terms in Eq.(\ref{XinY}) and (\ref{XoutY}) and
demonstrate that indeed $\mathbb E(X \circ Y | \sigma) \leq
\mathbb E(X | \sigma) \circ \mathbb E(Y | \sigma)$.  First note
that $\mathcal O(X \circ Y) = \mathcal O(X) \circ \mathcal O(Y)$.
Next, as discussed previously, we see that $A_i(X \circ Y) \subset (A_i(X) \circ \mathcal O(Y))
\cup (\mathcal O(X) \circ A_i(Y)), 1 \leq i \leq 3$.  Finally, and this is the key case, we claim that
$\mathcal F(X\circ Y) \subset A_1(X) \circ A_1(Y)$.  This follows
from the observation we made before, which is that if $\sigma \in
\mathcal F(X \circ Y)$, then we must see the alternating
configuration on the flower, requiring next to nearest neighbor
transmissions through the iris for both paths; such a $\sigma$
certainly lies in $A_1(X) \circ A_1(Y)$.  Thus we are done, assuming that $(a+s)^2 \geq s$ -- but this is equivalent to the statement that $a^2 \geq 2 s^2$.  

We have established the claim for the case of a single flower and
two paths.  Next we may induct on the number of flowers, as follows.
Suppose now the claim is established for $K-1$ flowers.  We can now
let $\sigma$ denote the configuration of all petals, filler, and
irises of the first $K-1$ flowers.  We condition on $\sigma$ as
above and adapt the notation so that the sets $\mathcal O$, $A_i$'s, and $\mathcal
F$ correspond to the $K^{\text{th}}$ flower.  The argument can then be
carried out exactly as above to yield the result for $K$ flowers and
two paths.  Finally we induct on the number of paths.  Suppose the
claim is true for $n-1$ paths.  Since the $\circ$ operation is
associative, we consider $(X_1 \circ \dots \circ X_{n-1}) \circ
X_n$, where the $X_i$'s are indicator functions of the $n$ paths. We
simply view $(X_1 \circ \dots \circ X_{n-1})$ as a single path--type
event and repeat the proof (note that the analogue of equation (\ref{XoutY}) may now contain non--trivial $\mathcal{F}$--type terms; these are immaterial since what is listed is already enough for an upper bound).  This argument is sufficient since no more than
two paths may share an iris under any circumstance.
\end{proof}

\subsection{On the Generalization of Cardy's Formula for $M(\partial \Omega) < 2$}\label{cardy}

Here we provide the necessary interior analyticity statement required to extract Cardy's Formula for the model in \cite{cardy} (the actual, full proof requires additional ingredients found in the companion work \cite{pt2}).  As described in $\S$\ref{model}, \cite{cardy} contains a proof of Cardy's formula for piecewise smooth domains, so what is needed here is a generalization to domains $\Omega$ with $M(\partial \Omega) < 2$.  What we will prove is the following:  
 
\begin{lemma}\label{cardy_formula'}
Let $\Omega$ denote any conformal triangular domain with $M(\partial \Omega) < 2$.  Let $u_{\varepsilon}^Y$, $v_{\varepsilon}^Y$ and $w_{\varepsilon}^Y$ denote the crossing probability functions as defined in $\Omega$ for the lattice at scale $\varepsilon$.  Then for the model as defined in  $\S$\ref{model}, we have 
\[\lim_{\varepsilon \rightarrow 0} u_\varepsilon^Y = u,\]
with similar results for $v_\varepsilon^Y$ and $w_\varepsilon^Y$ and the corresponding blue versions of these functions, where $u$, $v$ and $w$ are the Cardy--Carleson functions.    
\end{lemma}

To prove the current statement, we start by repeating the proof in \cite{cardy} up to Lemma 7.2 and Corollary 7.4 -- the one place where the assumption on a piecewise smooth boundary is used.  We now give a quick exposition of the (relevant portions of the) strategy of proof in \cite{cardy}.  The idea (directly inherited from \cite{stas_perc}) is to represent the derivative of the crossing probability functions as a ``three--arm'' event, e.g., two blue paths and one yellow path from some point to the boundaries, with all paths disjoint, and then derive Cauchy--Riemann type identities by switching the color of one of the arms.  

In order to accomplish this color switching in our model, it was necessary to introduce a \emph{stochastic} notion of disjointness.  This amounted to the introduction of a large class of random variables which indicate whether or not a percolation configuration contributes to the event of interest (e.g., a blue path from $\mathcal A$ to $\mathcal B$, separating $z$ from $\mathcal C$).  We call the restrictions and permissions given by these random variables \emph{$*$--rules}.  The $*$--rules may at times call a self--avoiding path illegitimate if it contains \emph{close encounters}, i.e., comes within one unit of itself; on the other hand, the $*$--rules may at other times permit a path which is not self--avoiding but in fact shares a hexagon.  Thus the $*$--rules are invoked only at shared hexagons and close encounter points of a path.  When a close encounter or sharing at a hexagon is required to achieve the desired path event it is called an \emph{essential lasso point}.

The fact that these $*$--rules may be implemented by random variables in a fashion which allows color switching is the content of Lemma 3.17 in \cite{cardy}.  The strategy was then to first prove that the $*$--version of e.g., the function $u_\varepsilon$, denoted $u_\varepsilon^*$, converges to $u$, then show that in the limit the starred and unstarred versions of the function coincide.  For the current work, the precise statement is as follows: 

\begin{lemma}\label{boundary}
Let $\Omega$ be a domain such that 
\[ M \equiv M(\partial \Omega) < 2.\]
Let $z$ denote a point in $\Omega$.   Consider the (blue version of the) function $u_\varepsilon(z)$ as defined in $\S$\ref{model}.  Let $u_\varepsilon^*(z)$ denote the version of $u_\varepsilon$ with the $*$-rules enforced.  Then,
\[ \lim_{\varepsilon \rightarrow \infty} |u_\varepsilon^*(z) - u_\varepsilon(z)| = 0.\]
In particular, on closed subsets of $\Omega$, the above is uniformly bounded by a constant times a power of $\varepsilon$.
\end{lemma}

Before we begin the proof we need some standard percolation notation.

\begin{defn}\label{perc}
Back on the unit hexagon lattice, if $L$ is a positive integer, let $B_L$ denote a box of side length $L$ centered at the origin.  Further, let $\Pi_5(L)$ denote the event of five disjoint paths, not all of the same color, starting from the origin and ending on $\partial B_L$.  Now let $m < n$ be positive integers, and let $\Pi(n, m)$ denote the event of five long arms, not all of the same color, connecting $\partial B_m$ and $\partial B_n$.  We use the notation $\pi_5(n)$ and $\pi_5(n, m)$ for the probabilities of $\Pi_5(n)$ and $\Pi_5(n, m)$, respectively.
\end{defn}

\begin{proof}[Proof of Lemma \ref{boundary}] We set $N = \varepsilon^{-1}$ and, without apology, we will denote the relevant functions by $u_{_N}$.  For convenience we recap the proof of Lemma 7.2 in \cite{cardy} (with one minor modification).  Let us first consider the event which is contained in both the starred and unstarred versions of the $u$--function, namely the event of a self--avoiding, non--self--touching path separating $z$ from $\mathcal C$, etc.  We will denote the indicator function of this event by $\mathfrak U_{_N}^-$.  Similarly, let us define an event, whose indicator is $\mathfrak U_{_N}^{*+}$, that contains both the starred and unstarred versions: This is the event that a separating path of the required type exists, with no restrictions on self--touching, and is allowed to share hexagons provided that permissions are granted.  It is obvious that 
\begin{equation}\label{tuff}
\mathbb E[\mathfrak U_{_N}^{*+} - \mathfrak U_{_N}^-] \geq |u_{_N}^*  - u_{_N}|.
\end{equation}

We turn to a description of the configurations, technically on $(\omega, X)$ (the enlarged probability space which include the permissions), for which $\mathfrak U_{_N}^{*+} = 1$ while $\mathfrak U_{_N}^- = 0$.  In such a configuration, the only separating paths contain an \emph{essential} lasso point which, we remind the reader, could be either a shared hexagon or a closed encounter pair.  Let us specify the lasso point under study to be the last such point on the journey from $\mathcal A$ to $\mathcal B$ (i.e., immediately after leaving this point, the path must capture $z$ without any further sharing or self--touching, then return to this point and continue on to $\mathcal B$). For standing notation, we denote this ``point'' by $z_0$.  A variety of paths converge at $z_0$: certainly there is a blue path from $\mathcal A$, denoted $B_{\mathcal A}$, a blue path to $\mathcal B$, denoted $B_{\mathcal B}$, and an additional loop starting from $z_0$ (or its immediate vicinity) which contains $z$ in its interior.  The loop we may view as two blue paths of comparable lengths, denoted $L_z^1$ and $L_z^2$.  However, since the lasso point was deemed to be essential, there are two additional yellow arms emanating from the immediate vicinity of $z_0$.  These yellow arms may themselves encircle the blue loop and/or terminate at the boundary $\mathcal C$.  We denote these yellow paths $Y_\mathcal C^1$ and $Y_\mathcal C^2$.  

Since $z_0$ is the last lasso point on the blue journey from $\mathcal A$ to $\mathcal B$, we automatically get that the two loop arms are \emph{strictly} self--avoiding.  Also, without loss of generality, we may take the yellow arms to be strictly self--avoiding.  Further, by Lemma 4.3 of \cite{cardy}, we may take either the portion of the path from $\mathcal A$ to $z_0$ to be strictly self--avoiding or the portion of the path from $\mathcal B$ to $z_0$ to be strictly self--avoiding.  To summarize, we have six paths emanating from $z_0$, four blue and two yellow, with all paths disjoint except for possible sharings between $B_{\mathcal A}$ and $B_{\mathcal B}$.  For simplicity, let us start with the connected component of $z$ in $\Omega \setminus (\alpha_k \cup \beta_k \cup \gamma_k)$ where $\alpha_k, \beta_k, \gamma_k$ are short crosscuts defining the prime ends $a, b, c$, respectively.  It is noted that in this restricted setting, the various portions of the boundary are at a finite (macroscopic) distance from one another.  Thus, on a mesoscopic scale, we are always near only a single boundary.  

The case where $z_0$ is close to $z$ is handled by RSW-type bounds (see proof of Lemma 7.2 in \cite{cardy}). The terms where $z_0$ is in the interior follow from the $5^+$ arm estimates; these arguments are the subject of Lemma 7.2 and Lemma 7.3 in \cite{cardy}.  We are left with the case where say $z_0$ is within a distance $N^\lambda$ of the boundary but outside some box of side $N^{\mu_2}$ separating $c$ from $z$.  

Let $\delta > 0$.  For $N$ large enough, $\partial \Omega$ can be covered by no more than $J_\delta N^{M+\delta - \lambda}$ boxes of side $N^\lambda$.  Now we take these boxes and expand by a factor of, say, two and we see that the region within $N^\lambda$ of the boundary can be covered by $J_\delta N^{M+\delta - \lambda}$ boxes of side $2N^\lambda$.  We surround each of these boxes by a box of side $N^{\mu_1}$, where $\mu_2 > \mu_1 > \lambda$.  

Now suppose $z_0$ is inside the inner box.  We still have the six arms $B_\mathcal A$, $B_\mathcal B$, $L_z^1$, $L_z^2$, $Y_\mathcal C^1$ and $Y_\mathcal C^2$, but since $z_0$ is now close to some boundary, we expect some arm(s) to be short (i.e., shorter than $N^\lambda$).  We note that the box of side $\mu_1$ is still away from $c$, and therefore we cannot have more than one of $B_\mathcal A$ and $B_\mathcal B$ be short due to being close to the boundary.  Also, since $z$ must be a distance of order $N$ away from the boundary, $z$ is outside of both of these boxes and therefore both $L_z^1$ and $L_z^2$ are long.  The upshot is that regardless of which boundary $z_0$ is close to, one and only one of the six arms will be short: If $z_0$ is close to $\mathcal A$ (respectively $\mathcal B$), then $B_\mathcal A$ (respectively $B_\mathcal B$) will be short, and if $z_0$ is close to $\mathcal C$, then a moment's reflection will show that only one of the yellow arms will be short.  

What we have is then five long arms and one short arm emanating from the immediate vicinity of $z_0$, and these arms either end on some boundary or the boundary of the outer box of side $N^{\mu_1}$.  For reasons which will momentarily become clear, we will now perform a color switch.  Topologically, the two yellow arms separate $L_z^1$ and $L_z^2$ from $B_\mathcal A$ and $B_\mathcal B$.  Denote the outer box by $B_{\mu_1}$ and consider now the region $T\equiv\Omega \cap B_{\mu_1}$.  The two yellow arms together form a ``crosscut'' (in the sense of Kesten \cite{ksz}) of $T$ .  This crosscut separates $T$ into two disjoint regions $T_b$ and $T_l$, where $T_b$ contains $B_{\mathcal A}$ and $B_{\mathcal B}$ and $T_l$ contains $L_z^1$ and $L_z^2$.  We condition on the crosscut which minimizes the area of $T_l$.  Next we apply Lemma 4.3 of \cite{cardy} to reduce the blue arm adjacent to the longer of the two yellow arms -- which we take to be $Y_\mathcal C^1$ -- to be strictly self--avoiding, which without loss of generality we assume to be $B_\mathcal A$.  Since $B_\mathcal A$ forms a crosscut of $T_b$, there is a crosscut which maximizes the region which contains $B_\mathcal B$, which we denote $\Omega_B$.  The region $\Omega_B$ is now an unconditioned region, and we may apply Lemma 3.17 of \cite{cardy} to switch the color of $B_\mathcal B$ from blue to yellow, while preserving the probability.  The resulting yellow path we will denote $Y_\mathcal B$.  

We now have three blue paths and three yellow paths.  The blue paths are now all strictly self--avoiding.  $Y_\mathcal C^1$ is still strictly self--avoiding, but the path $Y_\mathcal B$ may very well interact with (i.e., share hexagons with, due to the $*$-rules) $Y_\mathcal C^2$.  If indeed there is sharing, then let $\hat{Y} = Y_\mathcal B \cup Y_\mathcal C^2$ be the geometric union of the two paths.  $\hat{Y}$ can then be reduced to be a strictly self--avoiding path, which we now denote $Y$.  In any case, we now have (at least) five long paths emanating from $z_0$, three blue and two yellow, with the yellow paths separating the blue paths, and with all paths strictly self--avoiding.  The probability of such an event is certainly bounded above (possibly strictly since the boxes will most likely intersect $\Omega^c$) by the full space event $\Pi_{5}(N^{\mu_1}, 2N^\lambda)$ -- see Definition \ref{perc}.  The upshot of Lemma 5 of \cite{ksz} is that 
\begin{equation}\label{five}\pi_{5}(N^{\mu_1}, 2N^\lambda) \leq C \left( \frac{N^\lambda}{N^{\mu_1}}\right)^2,\end{equation}
where $C$ is a constant.  This result can, almost without modification, be taken verbatim from \cite{ksz}; the proviso therein which concerned ``relocation of arms'' was discussed in the first paragraph of the proof of Lemma 7.3 in \cite{cardy}.  We consider (\ref{five}) to be established.

If we sum over all such boxes of side $2N^\lambda$, we find that the contribution from the near boundary regions is a constant times 
\[N^{M+\delta - \lambda + 2\lambda - 2\mu_1} = N^{M + \delta + \lambda - 2\mu_1}.\]  
Since $M < 2$, we may first choose $\delta$ and $\lambda$ such that $M + \delta + \lambda < 2$, and next we will choose $\mu_2$ and then $\mu_1$ large enough so that the exponent is negative.  

Finally let us take care of the crosscuts.  We shall show that for large $k$, the event that a path emanates from the crosscut e.g., $\beta_k$ and goes to $\mathcal B$ tends to 0 as $k \rightarrow \infty$ (uniformly in $N$ for all $N$ sufficiently large): Indeed, although the prime end $b$ may be a continua, the probability of a path emanating from $b$ is ``as small'' as though $b$ were a point.  Let us begin by looking at the conformal rectangle $B_k \setminus B_{2k}$ defined by the relevant crosscuts.  We now mollify $B_k \setminus B_{2k}$ so that the resulting domain has smooth boundary and lies strictly in $\Omega$: This is easily accomplished by deleting from $B_k \setminus B_{2k}$ the image under the conformal map $\phi: \mathbb H \rightarrow \Omega$ of some $\delta$ neighborhood of $\partial \mathbb H$, where $\delta_k > 0$ is chosen so small that the said image is within some (Euclidean distance) $\eta_k$ of $\partial \Omega$.  Let us denote the resulting domain by $R_k$.  Since $R_k$ has smooth boundary, the result of \cite{cardy} applies and we may apply Cardy's Formula inside $R_k$ to see that the probability of a ``lateral'' yellow crossing (i.e., one ``parallel'' to $\beta_k$ and $\beta_{2k}$) is uniformly bounded from below, independently of $k$, if $\eta_k$ is properly chosen.  We may even assume that the crossing takes place in the ``bottom'' half of $R_k$, which will allow us to construct Harris annuli of order $\eta_k$ enabling a connection to the actual boundary.  Thus, having achieved all this, looking at the lowest such crossing, we may RSW continue the crossing to the actual $\partial \Omega$, with probability uniformly bounded from below.  It is now straightforward to observe that in the presence of such a yellow crossing, no blue path may emanate from $\beta_k$.  Performing this construction on a multitude of scales, it is clear, as $\varepsilon \to 0$ that with probability tending to one, no blue path emanates from this prime end.
%

All estimates described above are uniform in $z$ provided $z$ remains a fixed non--zero (Euclidean) distance from the boundary.  And, finally, the proof of Lemma \ref{boundary} for $v_N$ and $w_N$ are the same. 
\end{proof}

\begin{proof}[Proof of Lemma \ref{cardy_formula'}]
Corollary 7.4 of \cite{cardy} concerned the difference between the blue and yellow versions of these functions (Cauchy--Riemann relations are only established for color--neutral sums).  However, the argument of Corollary 7.4 in \cite{cardy} reduced the difference between the two colored versions to six arm events in the bulk and five arm events near the boundary, to which the above arguments can be applied. 
Replacing Lemma 7.2 (and Lemma 7.3) in \cite{cardy} with Lemma \ref{boundary} gives a proof of Lemma \ref{cardy_formula'}.  
\end{proof}

\section*{\large{Acknowledgments}} 

\hspace{16 pt}The authors are grateful to the IPAM institute at UCLA for their hospitality and support during the ÒRandom Shapes ConferenceÓ (where this work began).  The conference was funded by the NSF under the grant DMS-0439872.  I.~B.~was partially supported by the NSERC under the DISCOVER grant 5810-2004-298433. L.~C.~was supported by the NSF under the grants DMS-0306167 and DMS-0805486.  H.~K.~L was supported by the NSF VIGRE grant, and the Graduate Research Mentorship Program and the Dissertation Year Fellowship Program at UCLA.

The authors would like to thank John Garnett for various pertinent conversations and would like to express their gratitude to Stas Smirnov for \emph{numerous} conversations and consultations.  The authors would also like to thank Wendelin Werner for some useful discussions.

\end{document}

%% file: triple_visit.tex
\begin{pgfpicture}{0bp}{0bp}{715.0bp}{575.0bp}
\begin{pgfscope}
\pgfsetlinewidth{0.5669299960136414bp}
\pgfsetrectcap 
\pgfsetmiterjoin \pgfsetmiterlimit{10.0}
\pgfpathmoveto{\pgfpoint{424.592106bp}{322.999933bp}}
\pgfpathcurveto{\pgfpoint{424.592106bp}{260.49333bp}}{\pgfpoint{374.561866bp}{209.821704bp}}{\pgfpoint{312.846485bp}{209.821704bp}}
\pgfpathcurveto{\pgfpoint{251.131104bp}{209.821704bp}}{\pgfpoint{201.100898bp}{260.49333bp}}{\pgfpoint{201.100898bp}{322.999933bp}}
\pgfpathcurveto{\pgfpoint{201.100898bp}{385.506537bp}}{\pgfpoint{251.131104bp}{436.178145bp}}{\pgfpoint{312.846485bp}{436.178145bp}}
\pgfpathcurveto{\pgfpoint{374.561866bp}{436.178145bp}}{\pgfpoint{424.592106bp}{385.506537bp}}{\pgfpoint{424.592106bp}{322.999933bp}}
\pgfclosepath
\color[rgb]{0.0,0.0,0.0}
\pgfusepath{stroke}
\end{pgfscope}
\begin{pgfscope}
\pgfsetlinewidth{0.5669299960136414bp}
\pgfsetrectcap 
\pgfsetmiterjoin \pgfsetmiterlimit{10.0}
\pgfpathmoveto{\pgfpoint{363.125366bp}{325.862889bp}}
\pgfpathcurveto{\pgfpoint{363.125366bp}{296.562918bp}}{\pgfpoint{339.673691bp}{272.810594bp}}{\pgfpoint{310.744606bp}{272.810594bp}}
\pgfpathcurveto{\pgfpoint{281.815521bp}{272.810594bp}}{\pgfpoint{258.363862bp}{296.562918bp}}{\pgfpoint{258.363862bp}{325.862889bp}}
\pgfpathcurveto{\pgfpoint{258.363862bp}{355.162859bp}}{\pgfpoint{281.815521bp}{378.915176bp}}{\pgfpoint{310.744606bp}{378.915176bp}}
\pgfpathcurveto{\pgfpoint{339.673691bp}{378.915176bp}}{\pgfpoint{363.125366bp}{355.162859bp}}{\pgfpoint{363.125366bp}{325.862889bp}}
\pgfclosepath
\color[rgb]{0.0,0.0,0.0}
\pgfusepath{stroke}
\end{pgfscope}
\begin{pgfscope}
\pgfsetlinewidth{4.0bp}
\pgfsetrectcap 
\pgfsetmiterjoin \pgfsetmiterlimit{10.0}
\pgfsetdash{{1.0bp}{1.0bp}}{1.0bp}
\pgfpathmoveto{\pgfpoint{89.938622bp}{119.459906bp}}
\pgfpathcurveto{\pgfpoint{117.636245bp}{120.892542bp}}{\pgfpoint{153.133774bp}{128.05572bp}}{\pgfpoint{177.806944bp}{136.651534bp}}
\pgfpathcurveto{\pgfpoint{202.480115bp}{145.247348bp}}{\pgfpoint{223.332923bp}{157.345161bp}}{\pgfpoint{240.842915bp}{175.810243bp}}
\pgfpathcurveto{\pgfpoint{260.263087bp}{193.320235bp}}{\pgfpoint{273.475172bp}{210.830227bp}}{\pgfpoint{281.911805bp}{227.385128bp}}
\pgfpathcurveto{\pgfpoint{290.348437bp}{243.940029bp}}{\pgfpoint{293.532072bp}{257.311296bp}}{\pgfpoint{300.058523bp}{271.319289bp}}
\pgfpathcurveto{\pgfpoint{306.584975bp}{285.327282bp}}{\pgfpoint{309.609428bp}{297.265913bp}}{\pgfpoint{303.878885bp}{303.792365bp}}
\pgfpathcurveto{\pgfpoint{298.148342bp}{310.318816bp}}{\pgfpoint{290.348437bp}{317.004449bp}}{\pgfpoint{276.181262bp}{317.163631bp}}
\pgfpathcurveto{\pgfpoint{262.014087bp}{317.322813bp}}{\pgfpoint{246.414276bp}{310.000453bp}}{\pgfpoint{227.471648bp}{303.792365bp}}
\pgfpathcurveto{\pgfpoint{204.708659bp}{294.719005bp}}{\pgfpoint{192.92921bp}{288.988463bp}}{\pgfpoint{161.570407bp}{296.151641bp}}
\pgfpathcurveto{\pgfpoint{141.672689bp}{306.180091bp}}{\pgfpoint{118.272972bp}{317.322813bp}}{\pgfpoint{106.17516bp}{333.400169bp}}
\pgfpathcurveto{\pgfpoint{95.987528bp}{347.567344bp}}{\pgfpoint{100.762981bp}{369.534425bp}}{\pgfpoint{106.17516bp}{384.975054bp}}
\pgfpathcurveto{\pgfpoint{111.587339bp}{400.415683bp}}{\pgfpoint{123.844333bp}{413.309404bp}}{\pgfpoint{142.468597bp}{414.582858bp}}
\pgfpathcurveto{\pgfpoint{161.092861bp}{418.721583bp}}{\pgfpoint{182.90076bp}{417.448129bp}}{\pgfpoint{193.088392bp}{411.717587bp}}
\pgfpathcurveto{\pgfpoint{203.276023bp}{406.942134bp}}{\pgfpoint{216.328926bp}{402.166682bp}}{\pgfpoint{228.426739bp}{390.705597bp}}
\pgfpathcurveto{\pgfpoint{240.524551bp}{382.109782bp}}{\pgfpoint{248.961184bp}{375.264967bp}}{\pgfpoint{257.079453bp}{366.828335bp}}
\pgfpathcurveto{\pgfpoint{265.197722bp}{358.391703bp}}{\pgfpoint{269.973174bp}{354.889704bp}}{\pgfpoint{282.866895bp}{346.771435bp}}
\pgfpathcurveto{\pgfpoint{295.760616bp}{338.653166bp}}{\pgfpoint{311.997154bp}{328.94308bp}}{\pgfpoint{320.115423bp}{344.861254bp}}
\pgfpathcurveto{\pgfpoint{333.009144bp}{363.6447bp}}{\pgfpoint{306.903338bp}{385.770963bp}}{\pgfpoint{301.968704bp}{398.34632bp}}
\pgfpathcurveto{\pgfpoint{276.02208bp}{427.158216bp}}{\pgfpoint{293.213708bp}{477.618828bp}}{\pgfpoint{329.666328bp}{498.630819bp}}
\pgfpathcurveto{\pgfpoint{345.106957bp}{507.226633bp}}{\pgfpoint{403.049111bp}{535.24262bp}}{\pgfpoint{445.232273bp}{533.969166bp}}
\pgfpathcurveto{\pgfpoint{480.729802bp}{530.785531bp}}{\pgfpoint{593.908022bp}{523.940716bp}}{\pgfpoint{612.373104bp}{450.876296bp}}
\pgfpathcurveto{\pgfpoint{628.928005bp}{409.329861bp}}{\pgfpoint{622.08319bp}{352.501978bp}}{\pgfpoint{600.912018bp}{320.028903bp}}
\pgfpathcurveto{\pgfpoint{569.234851bp}{288.510917bp}}{\pgfpoint{538.03523bp}{277.049832bp}}{\pgfpoint{512.088606bp}{276.094741bp}}
\pgfpathcurveto{\pgfpoint{499.513248bp}{273.22947bp}}{\pgfpoint{463.538174bp}{283.735465bp}}{\pgfpoint{407.028655bp}{300.927093bp}}
\pgfpathcurveto{\pgfpoint{372.486217bp}{307.612727bp}}{\pgfpoint{333.009144bp}{320.347266bp}}{\pgfpoint{327.756147bp}{308.567817bp}}
\pgfpathcurveto{\pgfpoint{322.503149bp}{296.788368bp}}{\pgfpoint{324.254148bp}{281.347739bp}}{\pgfpoint{331.576509bp}{268.454018bp}}
\pgfpathcurveto{\pgfpoint{338.898869bp}{255.560296bp}}{\pgfpoint{353.225226bp}{239.005395bp}}{\pgfpoint{374.555579bp}{218.789314bp}}
\pgfpathcurveto{\pgfpoint{395.885933bp}{198.573232bp}}{\pgfpoint{424.220283bp}{174.695971bp}}{\pgfpoint{459.55863bp}{147.157529bp}}
\pgfpathcurveto{\pgfpoint{459.55863bp}{147.157529bp}}{\pgfpoint{459.55863bp}{147.157529bp}}{\pgfpoint{459.55863bp}{147.157529bp}}
\pgfpathcurveto{\pgfpoint{459.55863bp}{147.157529bp}}{\pgfpoint{459.55863bp}{147.157529bp}}{\pgfpoint{459.55863bp}{147.157529bp}}
\color[rgb]{0.5,0.0,0.30000001192092896}
\pgfusepath{stroke}
\end{pgfscope}
\begin{pgfscope}
\pgftransformcm{2.97432}{-0.0}{0.0}{2.97432}{\pgfpoint{69.364064bp}{374.295427bp}}
\pgftext[left,base]{\sffamily\mdseries\upshape\normalsize
\color[rgb]{0.0,0.0,0.0}$t_b$}
\end{pgfscope}
\begin{pgfscope}
\pgftransformcm{2.97432}{-0.0}{0.0}{2.97432}{\pgfpoint{181.082063bp}{110.117153bp}}
\pgftext[left,base]{\sffamily\mdseries\upshape\normalsize
\color[rgb]{0.0,0.0,0.0}$t_a$}
\end{pgfscope}
\begin{pgfscope}
\pgftransformcm{2.97432}{-0.0}{0.0}{2.97432}{\pgfpoint{609.372437bp}{481.574586bp}}
\pgftext[left,base]{\sffamily\mdseries\upshape\normalsize
\color[rgb]{0.0,0.0,0.0}$t_c$}
\end{pgfscope}
\begin{pgfscope}
\pgftransformcm{2.97432}{-0.0}{0.0}{2.97432}{\pgfpoint{397.752202bp}{143.51217bp}}
\pgftext[left,base]{\sffamily\mdseries\upshape\normalsize
\color[rgb]{0.0,0.0,0.0}$t_d$}
\end{pgfscope}
\begin{pgfscope}
\pgftransformcm{1.48716}{-0.0}{0.0}{1.48716}{\pgfpoint{280.685923bp}{291.64698bp}}
\pgftext[left,base]{\sffamily\mdseries\upshape\large
\color[rgb]{0.0,0.0,0.0}$t_1$}
\end{pgfscope}
\begin{pgfscope}
\pgftransformcm{1.48716}{-0.0}{0.0}{1.48716}{\pgfpoint{302.10776bp}{353.318013bp}}
\pgftext[left,base]{\sffamily\mdseries\upshape\large
\color[rgb]{0.0,0.0,0.0}$t_2$}
\end{pgfscope}
\begin{pgfscope}
\pgftransformcm{1.48716}{-0.0}{0.0}{1.48716}{\pgfpoint{337.954746bp}{296.360246bp}}
\pgftext[left,base]{\sffamily\mdseries\upshape\large
\color[rgb]{0.0,0.0,0.0}$t_3$}
\end{pgfscope}
\begin{pgfscope}
\pgfsetlinewidth{0.5669299960136414bp}
\pgfsetrectcap 
\pgfsetmiterjoin \pgfsetmiterlimit{10.0}
\pgfpathmoveto{\pgfpoint{269.704555bp}{93.558026bp}}
\pgfpathlineto{\pgfpoint{298.896746bp}{93.558026bp}}
\pgfpathlineto{\pgfpoint{298.896746bp}{93.558026bp}}
\pgfpathlineto{\pgfpoint{299.051202bp}{93.558026bp}}
\pgfpathlineto{\pgfpoint{269.704555bp}{93.558026bp}}
\pgfclosepath
\color[rgb]{0.0,0.0,0.0}
\pgfusepath{stroke}
\end{pgfscope}
\begin{pgfscope}
\pgfsetlinewidth{0.5669299960136414bp}
\pgfsetrectcap 
\pgfsetmiterjoin \pgfsetmiterlimit{10.0}
\pgfpathmoveto{\pgfpoint{268.959835bp}{100.592872bp}}
\pgfpathlineto{\pgfpoint{268.959835bp}{86.505288bp}}
\pgfpathlineto{\pgfpoint{268.959835bp}{86.505288bp}}
\color[rgb]{0.0,0.0,0.0}
\pgfusepath{stroke}
\end{pgfscope}
\begin{pgfscope}
\pgfsetlinewidth{0.5669299960136414bp}
\pgfsetrectcap 
\pgfsetmiterjoin \pgfsetmiterlimit{10.0}
\pgfpathmoveto{\pgfpoint{368.69909bp}{100.183025bp}}
\pgfpathlineto{\pgfpoint{368.69909bp}{86.09544bp}}
\pgfpathlineto{\pgfpoint{368.69909bp}{86.09544bp}}
\color[rgb]{0.0,0.0,0.0}
\pgfusepath{stroke}
\end{pgfscope}
\begin{pgfscope}
\pgfsetlinewidth{0.5669299960136414bp}
\pgfsetrectcap 
\pgfsetmiterjoin \pgfsetmiterlimit{10.0}
\pgfpathmoveto{\pgfpoint{419.4521bp}{44.631327bp}}
\pgfpathlineto{\pgfpoint{419.4521bp}{30.543742bp}}
\pgfpathlineto{\pgfpoint{419.4521bp}{30.543742bp}}
\color[rgb]{0.0,0.0,0.0}
\pgfusepath{stroke}
\end{pgfscope}
\begin{pgfscope}
\pgfsetlinewidth{0.5669299960136414bp}
\pgfsetrectcap 
\pgfsetmiterjoin \pgfsetmiterlimit{10.0}
\pgfpathmoveto{\pgfpoint{223.18101bp}{44.153781bp}}
\pgfpathlineto{\pgfpoint{223.18101bp}{30.066197bp}}
\pgfpathlineto{\pgfpoint{223.18101bp}{30.066197bp}}
\color[rgb]{0.0,0.0,0.0}
\pgfusepath{stroke}
\end{pgfscope}
\begin{pgfscope}
\pgftransformcm{1.48716}{-0.0}{0.0}{1.48716}{\pgfpoint{305.604564bp}{86.986667bp}}
\pgftext[left,base]{\sffamily\mdseries\upshape\Large
\color[rgb]{0.0,0.0,0.0}2$\Delta_1$}
\end{pgfscope}
\begin{pgfscope}
\pgftransformcm{1.48716}{-0.0}{0.0}{1.48716}{\pgfpoint{307.829132bp}{31.813723bp}}
\pgftext[left,base]{\sffamily\mdseries\upshape\Large
\color[rgb]{0.0,0.0,0.0}2$\Delta_2$}
\end{pgfscope}
\begin{pgfscope}
\pgfsetlinewidth{0.5669299960136414bp}
\pgfsetrectcap 
\pgfsetmiterjoin \pgfsetmiterlimit{10.0}
\pgfpathmoveto{\pgfpoint{341.828568bp}{93.176526bp}}
\pgfpathlineto{\pgfpoint{368.10154bp}{93.176526bp}}
\pgfpathlineto{\pgfpoint{368.10154bp}{93.176526bp}}
\pgfpathlineto{\pgfpoint{368.24055bp}{93.176526bp}}
\pgfpathlineto{\pgfpoint{341.828568bp}{93.176526bp}}
\pgfclosepath
\color[rgb]{0.0,0.0,0.0}
\pgfusepath{stroke}
\end{pgfscope}
\begin{pgfscope}
\pgfsetlinewidth{0.5669299960136414bp}
\pgfsetrectcap 
\pgfsetmiterjoin \pgfsetmiterlimit{10.0}
\pgfpathmoveto{\pgfpoint{348.155191bp}{37.233664bp}}
\pgfpathlineto{\pgfpoint{418.800292bp}{37.233664bp}}
\pgfpathlineto{\pgfpoint{418.800292bp}{37.233664bp}}
\pgfpathlineto{\pgfpoint{419.174076bp}{37.233664bp}}
\pgfpathlineto{\pgfpoint{348.155191bp}{37.233664bp}}
\pgfclosepath
\color[rgb]{0.0,0.0,0.0}
\pgfusepath{stroke}
\end{pgfscope}
\begin{pgfscope}
\pgfsetlinewidth{0.5669299960136414bp}
\pgfsetrectcap 
\pgfsetmiterjoin \pgfsetmiterlimit{10.0}
\pgfpathmoveto{\pgfpoint{223.993429bp}{37.233664bp}}
\pgfpathlineto{\pgfpoint{294.63853bp}{37.233664bp}}
\pgfpathlineto{\pgfpoint{294.63853bp}{37.233664bp}}
\pgfpathlineto{\pgfpoint{295.012314bp}{37.233664bp}}
\pgfpathlineto{\pgfpoint{223.993429bp}{37.233664bp}}
\pgfclosepath
\color[rgb]{0.0,0.0,0.0}
\pgfusepath{stroke}
\end{pgfscope}
\begin{pgfscope}
\pgfsetlinewidth{2.0bp}
\pgfsetrectcap 
\pgfsetmiterjoin \pgfsetmiterlimit{10.0}
\pgfpathmoveto{\pgfpoint{112.553948bp}{386.400313bp}}
\pgfpathcurveto{\pgfpoint{112.553948bp}{382.762273bp}}{\pgfpoint{109.759954bp}{379.813059bp}}{\pgfpoint{106.313387bp}{379.813059bp}}
\pgfpathcurveto{\pgfpoint{102.86682bp}{379.813059bp}}{\pgfpoint{100.072826bp}{382.762273bp}}{\pgfpoint{100.072826bp}{386.400313bp}}
\pgfpathcurveto{\pgfpoint{100.072826bp}{390.038354bp}}{\pgfpoint{102.86682bp}{392.987576bp}}{\pgfpoint{106.313387bp}{392.987576bp}}
\pgfpathcurveto{\pgfpoint{109.759954bp}{392.987576bp}}{\pgfpoint{112.553948bp}{390.038354bp}}{\pgfpoint{112.553948bp}{386.400313bp}}
\pgfclosepath
\color[rgb]{1.0,0.85,0.0}\pgfseteorule\pgfusepath{fill}
\pgfpathmoveto{\pgfpoint{112.553948bp}{386.400313bp}}
\pgfpathcurveto{\pgfpoint{112.553948bp}{382.762273bp}}{\pgfpoint{109.759954bp}{379.813059bp}}{\pgfpoint{106.313387bp}{379.813059bp}}
\pgfpathcurveto{\pgfpoint{102.86682bp}{379.813059bp}}{\pgfpoint{100.072826bp}{382.762273bp}}{\pgfpoint{100.072826bp}{386.400313bp}}
\pgfpathcurveto{\pgfpoint{100.072826bp}{390.038354bp}}{\pgfpoint{102.86682bp}{392.987576bp}}{\pgfpoint{106.313387bp}{392.987576bp}}
\pgfpathcurveto{\pgfpoint{109.759954bp}{392.987576bp}}{\pgfpoint{112.553948bp}{390.038354bp}}{\pgfpoint{112.553948bp}{386.400313bp}}
\pgfclosepath
\color[rgb]{0.0,0.0,0.0}
\pgfusepath{stroke}
\end{pgfscope}
\begin{pgfscope}
\pgfsetlinewidth{2.0bp}
\pgfsetrectcap 
\pgfsetmiterjoin \pgfsetmiterlimit{10.0}
\pgfpathmoveto{\pgfpoint{303.981888bp}{309.583229bp}}
\pgfpathcurveto{\pgfpoint{303.981888bp}{305.945189bp}}{\pgfpoint{301.187893bp}{302.995975bp}}{\pgfpoint{297.741326bp}{302.995975bp}}
\pgfpathcurveto{\pgfpoint{294.29476bp}{302.995975bp}}{\pgfpoint{291.500765bp}{305.945189bp}}{\pgfpoint{291.500765bp}{309.583229bp}}
\pgfpathcurveto{\pgfpoint{291.500765bp}{313.22127bp}}{\pgfpoint{294.29476bp}{316.170492bp}}{\pgfpoint{297.741326bp}{316.170492bp}}
\pgfpathcurveto{\pgfpoint{301.187893bp}{316.170492bp}}{\pgfpoint{303.981888bp}{313.22127bp}}{\pgfpoint{303.981888bp}{309.583229bp}}
\pgfclosepath
\color[rgb]{1.0,0.85,0.0}\pgfseteorule\pgfusepath{fill}
\pgfpathmoveto{\pgfpoint{303.981888bp}{309.583229bp}}
\pgfpathcurveto{\pgfpoint{303.981888bp}{305.945189bp}}{\pgfpoint{301.187893bp}{302.995975bp}}{\pgfpoint{297.741326bp}{302.995975bp}}
\pgfpathcurveto{\pgfpoint{294.29476bp}{302.995975bp}}{\pgfpoint{291.500765bp}{305.945189bp}}{\pgfpoint{291.500765bp}{309.583229bp}}
\pgfpathcurveto{\pgfpoint{291.500765bp}{313.22127bp}}{\pgfpoint{294.29476bp}{316.170492bp}}{\pgfpoint{297.741326bp}{316.170492bp}}
\pgfpathcurveto{\pgfpoint{301.187893bp}{316.170492bp}}{\pgfpoint{303.981888bp}{313.22127bp}}{\pgfpoint{303.981888bp}{309.583229bp}}
\pgfclosepath
\color[rgb]{0.0,0.0,0.0}
\pgfusepath{stroke}
\end{pgfscope}
\begin{pgfscope}
\pgfsetlinewidth{2.0bp}
\pgfsetrectcap 
\pgfsetmiterjoin \pgfsetmiterlimit{10.0}
\pgfpathmoveto{\pgfpoint{319.263335bp}{339.191033bp}}
\pgfpathcurveto{\pgfpoint{319.263335bp}{335.552993bp}}{\pgfpoint{316.469341bp}{332.603779bp}}{\pgfpoint{313.022774bp}{332.603779bp}}
\pgfpathcurveto{\pgfpoint{309.576207bp}{332.603779bp}}{\pgfpoint{306.782213bp}{335.552993bp}}{\pgfpoint{306.782213bp}{339.191033bp}}
\pgfpathcurveto{\pgfpoint{306.782213bp}{342.829074bp}}{\pgfpoint{309.576207bp}{345.778296bp}}{\pgfpoint{313.022774bp}{345.778296bp}}
\pgfpathcurveto{\pgfpoint{316.469341bp}{345.778296bp}}{\pgfpoint{319.263335bp}{342.829074bp}}{\pgfpoint{319.263335bp}{339.191033bp}}
\pgfclosepath
\color[rgb]{1.0,0.85,0.0}\pgfseteorule\pgfusepath{fill}
\pgfpathmoveto{\pgfpoint{319.263335bp}{339.191033bp}}
\pgfpathcurveto{\pgfpoint{319.263335bp}{335.552993bp}}{\pgfpoint{316.469341bp}{332.603779bp}}{\pgfpoint{313.022774bp}{332.603779bp}}
\pgfpathcurveto{\pgfpoint{309.576207bp}{332.603779bp}}{\pgfpoint{306.782213bp}{335.552993bp}}{\pgfpoint{306.782213bp}{339.191033bp}}
\pgfpathcurveto{\pgfpoint{306.782213bp}{342.829074bp}}{\pgfpoint{309.576207bp}{345.778296bp}}{\pgfpoint{313.022774bp}{345.778296bp}}
\pgfpathcurveto{\pgfpoint{316.469341bp}{345.778296bp}}{\pgfpoint{319.263335bp}{342.829074bp}}{\pgfpoint{319.263335bp}{339.191033bp}}
\pgfclosepath
\color[rgb]{0.0,0.0,0.0}
\pgfusepath{stroke}
\end{pgfscope}
\begin{pgfscope}
\pgfsetlinewidth{2.0bp}
\pgfsetrectcap 
\pgfsetmiterjoin \pgfsetmiterlimit{10.0}
\pgfpathmoveto{\pgfpoint{335.499873bp}{309.583229bp}}
\pgfpathcurveto{\pgfpoint{335.499873bp}{305.945189bp}}{\pgfpoint{332.705878bp}{302.995975bp}}{\pgfpoint{329.259312bp}{302.995975bp}}
\pgfpathcurveto{\pgfpoint{325.812745bp}{302.995975bp}}{\pgfpoint{323.018751bp}{305.945189bp}}{\pgfpoint{323.018751bp}{309.583229bp}}
\pgfpathcurveto{\pgfpoint{323.018751bp}{313.22127bp}}{\pgfpoint{325.812745bp}{316.170492bp}}{\pgfpoint{329.259312bp}{316.170492bp}}
\pgfpathcurveto{\pgfpoint{332.705878bp}{316.170492bp}}{\pgfpoint{335.499873bp}{313.22127bp}}{\pgfpoint{335.499873bp}{309.583229bp}}
\pgfclosepath
\color[rgb]{1.0,0.85,0.0}\pgfseteorule\pgfusepath{fill}
\pgfpathmoveto{\pgfpoint{335.499873bp}{309.583229bp}}
\pgfpathcurveto{\pgfpoint{335.499873bp}{305.945189bp}}{\pgfpoint{332.705878bp}{302.995975bp}}{\pgfpoint{329.259312bp}{302.995975bp}}
\pgfpathcurveto{\pgfpoint{325.812745bp}{302.995975bp}}{\pgfpoint{323.018751bp}{305.945189bp}}{\pgfpoint{323.018751bp}{309.583229bp}}
\pgfpathcurveto{\pgfpoint{323.018751bp}{313.22127bp}}{\pgfpoint{325.812745bp}{316.170492bp}}{\pgfpoint{329.259312bp}{316.170492bp}}
\pgfpathcurveto{\pgfpoint{332.705878bp}{316.170492bp}}{\pgfpoint{335.499873bp}{313.22127bp}}{\pgfpoint{335.499873bp}{309.583229bp}}
\pgfclosepath
\color[rgb]{0.0,0.0,0.0}
\pgfusepath{stroke}
\end{pgfscope}
\begin{pgfscope}
\pgfsetlinewidth{2.0bp}
\pgfsetrectcap 
\pgfsetmiterjoin \pgfsetmiterlimit{10.0}
\pgfpathmoveto{\pgfpoint{599.10484bp}{485.319874bp}}
\pgfpathcurveto{\pgfpoint{599.10484bp}{481.681833bp}}{\pgfpoint{596.310846bp}{478.73262bp}}{\pgfpoint{592.864279bp}{478.73262bp}}
\pgfpathcurveto{\pgfpoint{589.417712bp}{478.73262bp}}{\pgfpoint{586.623718bp}{481.681833bp}}{\pgfpoint{586.623718bp}{485.319874bp}}
\pgfpathcurveto{\pgfpoint{586.623718bp}{488.957914bp}}{\pgfpoint{589.417712bp}{491.907136bp}}{\pgfpoint{592.864279bp}{491.907136bp}}
\pgfpathcurveto{\pgfpoint{596.310846bp}{491.907136bp}}{\pgfpoint{599.10484bp}{488.957914bp}}{\pgfpoint{599.10484bp}{485.319874bp}}
\pgfclosepath
\color[rgb]{1.0,0.85,0.0}\pgfseteorule\pgfusepath{fill}
\pgfpathmoveto{\pgfpoint{599.10484bp}{485.319874bp}}
\pgfpathcurveto{\pgfpoint{599.10484bp}{481.681833bp}}{\pgfpoint{596.310846bp}{478.73262bp}}{\pgfpoint{592.864279bp}{478.73262bp}}
\pgfpathcurveto{\pgfpoint{589.417712bp}{478.73262bp}}{\pgfpoint{586.623718bp}{481.681833bp}}{\pgfpoint{586.623718bp}{485.319874bp}}
\pgfpathcurveto{\pgfpoint{586.623718bp}{488.957914bp}}{\pgfpoint{589.417712bp}{491.907136bp}}{\pgfpoint{592.864279bp}{491.907136bp}}
\pgfpathcurveto{\pgfpoint{596.310846bp}{491.907136bp}}{\pgfpoint{599.10484bp}{488.957914bp}}{\pgfpoint{599.10484bp}{485.319874bp}}
\pgfclosepath
\color[rgb]{0.0,0.0,0.0}
\pgfusepath{stroke}
\end{pgfscope}
\begin{pgfscope}
\pgfsetlinewidth{2.0bp}
\pgfsetrectcap 
\pgfsetmiterjoin \pgfsetmiterlimit{10.0}
\pgfpathmoveto{\pgfpoint{198.921937bp}{141.487308bp}}
\pgfpathcurveto{\pgfpoint{198.921937bp}{137.849268bp}}{\pgfpoint{196.127943bp}{134.900054bp}}{\pgfpoint{192.681376bp}{134.900054bp}}
\pgfpathcurveto{\pgfpoint{189.234809bp}{134.900054bp}}{\pgfpoint{186.440815bp}{137.849268bp}}{\pgfpoint{186.440815bp}{141.487308bp}}
\pgfpathcurveto{\pgfpoint{186.440815bp}{145.125349bp}}{\pgfpoint{189.234809bp}{148.074571bp}}{\pgfpoint{192.681376bp}{148.074571bp}}
\pgfpathcurveto{\pgfpoint{196.127943bp}{148.074571bp}}{\pgfpoint{198.921937bp}{145.125349bp}}{\pgfpoint{198.921937bp}{141.487308bp}}
\pgfclosepath
\color[rgb]{1.0,0.85,0.0}\pgfseteorule\pgfusepath{fill}
\pgfpathmoveto{\pgfpoint{198.921937bp}{141.487308bp}}
\pgfpathcurveto{\pgfpoint{198.921937bp}{137.849268bp}}{\pgfpoint{196.127943bp}{134.900054bp}}{\pgfpoint{192.681376bp}{134.900054bp}}
\pgfpathcurveto{\pgfpoint{189.234809bp}{134.900054bp}}{\pgfpoint{186.440815bp}{137.849268bp}}{\pgfpoint{186.440815bp}{141.487308bp}}
\pgfpathcurveto{\pgfpoint{186.440815bp}{145.125349bp}}{\pgfpoint{189.234809bp}{148.074571bp}}{\pgfpoint{192.681376bp}{148.074571bp}}
\pgfpathcurveto{\pgfpoint{196.127943bp}{148.074571bp}}{\pgfpoint{198.921937bp}{145.125349bp}}{\pgfpoint{198.921937bp}{141.487308bp}}
\pgfclosepath
\color[rgb]{0.0,0.0,0.0}
\pgfusepath{stroke}
\end{pgfscope}
\begin{pgfscope}
\pgfsetlinewidth{2.0bp}
\pgfsetrectcap 
\pgfsetmiterjoin \pgfsetmiterlimit{10.0}
\pgfpathmoveto{\pgfpoint{416.682562bp}{186.37656bp}}
\pgfpathcurveto{\pgfpoint{416.682562bp}{182.738519bp}}{\pgfpoint{413.888568bp}{179.789306bp}}{\pgfpoint{410.442001bp}{179.789306bp}}
\pgfpathcurveto{\pgfpoint{406.995434bp}{179.789306bp}}{\pgfpoint{404.20144bp}{182.738519bp}}{\pgfpoint{404.20144bp}{186.37656bp}}
\pgfpathcurveto{\pgfpoint{404.20144bp}{190.0146bp}}{\pgfpoint{406.995434bp}{192.963822bp}}{\pgfpoint{410.442001bp}{192.963822bp}}
\pgfpathcurveto{\pgfpoint{413.888568bp}{192.963822bp}}{\pgfpoint{416.682562bp}{190.0146bp}}{\pgfpoint{416.682562bp}{186.37656bp}}
\pgfclosepath
\color[rgb]{1.0,0.85,0.0}\pgfseteorule\pgfusepath{fill}
\pgfpathmoveto{\pgfpoint{416.682562bp}{186.37656bp}}
\pgfpathcurveto{\pgfpoint{416.682562bp}{182.738519bp}}{\pgfpoint{413.888568bp}{179.789306bp}}{\pgfpoint{410.442001bp}{179.789306bp}}
\pgfpathcurveto{\pgfpoint{406.995434bp}{179.789306bp}}{\pgfpoint{404.20144bp}{182.738519bp}}{\pgfpoint{404.20144bp}{186.37656bp}}
\pgfpathcurveto{\pgfpoint{404.20144bp}{190.0146bp}}{\pgfpoint{406.995434bp}{192.963822bp}}{\pgfpoint{410.442001bp}{192.963822bp}}
\pgfpathcurveto{\pgfpoint{413.888568bp}{192.963822bp}}{\pgfpoint{416.682562bp}{190.0146bp}}{\pgfpoint{416.682562bp}{186.37656bp}}
\pgfclosepath
\color[rgb]{0.0,0.0,0.0}
\pgfusepath{stroke}
\end{pgfscope}
\begin{pgfscope}
\pgfsetlinewidth{1.0bp}
\pgfsetrectcap 
\pgfsetmiterjoin \pgfsetmiterlimit{10.0}
\pgfpathmoveto{\pgfpoint{0.500003bp}{574.5bp}}
\pgfpathlineto{\pgfpoint{0.500003bp}{0.5bp}}
\pgfpathlineto{\pgfpoint{714.500003bp}{0.5bp}}
\pgfpathlineto{\pgfpoint{714.500003bp}{574.5bp}}
\pgfpathlineto{\pgfpoint{0.500003bp}{574.5bp}}
\pgfclosepath
\color[rgb]{0.0,0.0,0.0}
\pgfusepath{stroke}
\end{pgfscope}
\end{pgfpicture}